\documentclass[a4paper,11pt]{article}
\title{On the Power of Conditional Samples in Distribution Testing\thanks{Research supported in part by an ERC-2007-StG grant number 202405.} \thanks{A preliminary version of this work appeared in the Proceedings of the 4th Innovations in Theoretical Computer Science conference (ITCS 2013)}}
\usepackage{amsmath}
\usepackage{amsthm}
\usepackage[noend]{algorithmic}
\usepackage{algorithm}
\algsetup{indent=2em}
\usepackage{amsfonts}
\usepackage{fullpage}
\usepackage{amssymb}

\author{
    Sourav Chakraborty\\
    Chennai Mathematical Institute,\\ India\\
    {\tt sourav@cmi.ac.in} \and Eldar Fischer\\
    Technion -- Israel Institute of Technology,\\ Haifa, Israel\\
    {\tt eldar@cs.technion.ac.il} \and Yonatan Goldhirsh\\
    Technion -- Israel Institute of Technology,\\ Haifa, Israel\\
    {\tt jongold@cs.technion.ac.il} \and Arie Matsliah\\
    Google Inc.,\\ Mountain View, California\\
    {\tt arie.matsliah@gmail.com}
    }

\newtheorem{theorem}{Theorem}[subsection]
\newtheorem{coro}[theorem]{Corollary}

\newtheorem{obs}[theorem]{Observation}
\newtheorem{lemma}[theorem]{Lemma}
\theoremstyle{definition}
\newtheorem{defin}[theorem]{Definition}
\newtheorem{algo}[theorem]{Algorithm}

\newcommand{\norm}[1]{\lVert{#1}\rVert}
\newcommand{\restrict}[2]{{#1}\upharpoonright_{#2}}
\newcommand{\bucketdist}[2]{{#1}_{\langle {#2}\rangle}}
\newcommand{\zo}{\{0,1\}}
\newcommand{\PP}{\mathcal{P}}
\newcommand{\poly}{\mathrm{poly}}
\newcommand{\naturals}{\mathbb{N}}
\newcommand{\reals}{\mathbb{R}}
\newcommand{\eat}[1]{}

\begin{document}

\maketitle
\thispagestyle{empty}
\begin{abstract}
In this paper we define and examine the power of the {\em conditional-sampling} oracle in the context of distribution-property testing. The conditional-sampling oracle for a discrete distribution $\mu$ takes as input a subset $S \subset [n]$ of the domain, and outputs a random sample $i \in S$ drawn according to $\mu$, conditioned on $S$ (and independently of all prior samples). The conditional-sampling oracle is a natural generalization of the ordinary sampling oracle, in which $S$ always equals $[n]$.

We show that with the conditional-sampling oracle, testing uniformity, testing identity to a known distribution, and testing any label-invariant property of distributions is easier than with the ordinary sampling oracle. On the other hand, we also show that for some distribution properties the sample-complexity remains near-maximal even with conditional sampling.
\end{abstract}

\newpage \pagenumbering{arabic}

\section{Introduction} In the last decade several works have investigated the
problem of testing various properties of huge data sets, that can be represented as
an unknown distribution from which independent samples can be taken. In {\em distribution-property testing}, the goal is to distinguish the case where the
samples come from a distribution that has a certain property $\PP$
from the case where the samples come from a
distribution that is far, in the variation distance, from any distribution
that has the property $\PP$ (the variation distance between two distributions $\mu$ and $\mu'$ over a common set $B$ is $\frac{1}{2}\sum_{i \in B}|\Pr_\mu [i]-\Pr_{\mu'}[i]|$, which is equal to the maximum difference in probability between the distributions for any possible event). In the traditional setting no access is provided to the distribution apart from the ability to take independent samples, and the two cases should be distinguished using as few samples as possible.

There are several natural distribution properties that were studied in this context: testing whether a distribution is uniform \cite{uni1}, testing identity between distributions (taking samples from both) \cite{BatuCloseness, collections}, testing whether a joint distribution is independent (a product of two distributions) \cite{Identity} and more. Some useful general techniques have also been designed to obtain nearly tight lower bounds on various distribution-property testing problems \cite{valiant}. Other tightly related works study the problems of estimating various measures of distributions, such as entropy \cite{ent1,ent2} or support size \cite{supp}.

Most attention has been given to testing properties of distributions over very large (discrete) domains, where the need for sublinear time and sample complexities is vital.
Distribution-property testers with a sublinear sample complexity are motivated by problems from various areas,
such as physics, cryptography, statistics, computational learning theory, property testing of graphs and sequences, and streaming algorithms (see the overview in \cite{collections} for a comprehensive list of references). Indeed, in many of the aforementioned works testers have been designed with sublinear sample (and time) complexity, that is often of the form $n^\alpha$, where $n$ is the size of the domain, and $\alpha$ is a positive constant smaller than $1$.

While most previous works are focused on the ordinary sampling oracle, other stronger oracles were considered too. A major reason is that the number of required samples, while sublinear, is still very large in the original model. The most notable example of a strong oracle is the one from \cite{ent1}, that also allows querying the exact probability weight of any element from the domain. Another research direction involved restricting the problem further, for example by adding the promise of the distribution being monotone \cite{monoprob}.

In this work we study the problem of testing several distribution properties in an unrestricted setting while providing for a stronger oracle, that can be thought of as more natural than the one of \cite{ent1} in some situations. Namely, we allow the samples obtained from the unknown distribution to be conditioned over specified subsets of the domain. In our setting, we assume that a sampling oracle to the unknown distribution $\mu$ over the discrete domain $[n]=\{1,\ldots,n\}$ is provided, that allows us to sample random (according to $\mu$) elements conditioned on any specified subset $S \subseteq [n]$. If the original distribution is described by the probabilities $p_1,\ldots,p_n$ (where the probability for obtaining $i\in [n]$ is $p_i$), then when restricting to $S$ the probability of sampling $i \in [n]$ is $p_i/(\sum_{j \in S}p_j)$ if $i \in S$ and $0$ otherwise (see the formal definition of the model and corresponding testers in Section \ref{sec:prelim}).

In various scenarios, conditional samples can be obtained naturally, or come at a low cost relative to that of extracting any sample -- see some illustrating examples in Section \ref{sec:examples}. This leads to the following natural question: can we reduce the sample complexity of distribution-property testers using conditional samples?

Indeed, conditional sampling is more powerful than the traditional model: We show that with conditional samples several natural distribution properties, such as uniformity,
can be tested with a constant number of samples (compared to $\widetilde\Theta(\sqrt{n})$ unconditional samples even for uniformity \cite{uni1,Identity}). The most general result of this paper (Section \ref{sec:label_invar}) is that any label-invariant property of distributions (a symmetric property in the terminology of \cite{valiant}) can be tested using $\poly(\log n)$ conditional samples.\footnote{We say that $f(\alpha_1,\ldots,\alpha_l)=\poly(g_1(\alpha_1,\ldots,\alpha_l),\ldots,g_k(\alpha_1,\ldots,\alpha_l))$ if there exists a polynomial $p(x_1,\ldots,x_k)$ such that $f\leq p(g_1,\ldots,g_k)$ for all values of $\alpha_1,\ldots,\alpha_l$ in their respective domains.}

On the other hand, there are properties for which testing remains almost as hard as possible even with conditional samples: We show a property of distributions that requires at least $\Omega(n)$ conditional samples to test (Section \ref{sec:lb}).

Another feature that makes conditional-samples interesting is that in contrast to the testers using ordinary samples, which are non-adaptive by definition, adaptivity (and the algorithmic aspect of testing) in conditional-sampling model plays an important role. For instance, the aforementioned task of testing uniformity, while still possible with a much better sampling complexity than in the traditional model, cannot be done non-adaptively with a constant number of samples (see Section \ref{sec:unilb}).

Before we move to some motivating examples, let us address the concern of whether arbitrary conditioning is realistic: While the examples below do relate to arbitrary conditioning, sometimes one would like the conditioning to be more restricted, in some sense describable by fewer than the $n$ bits required to describe the conditioning set $S$. In fact, many of our algorithms require less than that. For example, the adaptive uniformity test takes only unconditional samples and samples conditioned on a constant size set, so the description size per sample is in fact $O(\log n)$, as there are $n^{O(1)}$ possibilities. The adaptive general label invariant property tester takes only samples conditioned to dyadic intervals of $[n]$, so here the description size is $O(\log n)$ as well. The non-adaptive tests do require general conditioning, as they pick uniformly random sets of prescribed sizes.

\subsection{Some motivating examples} \label{sec:examples}
\subsubsection*{Lottery machines}
The {\em gravity pick lottery machine} is the most common lottery machine used worldwide to pick random numbers. A set $B$ of balls, each marked with a unique number $i \in \naturals$, are dropped into the machine while it is spinning, and after certain amount of time the machine allows a single ball to drop out. Ensuring that such a machine is fair is an important real-life problem.\footnote{As was demonstrated in the the Pennsylvania Lottery scandal, see e.g.\\ {\tt http://en.wikipedia.org/w/index.php?title=1980\_Pennsylvania\_Lottery\_scandal\&oldid=496671681}}

Suppose that, given a machine and set of balls, we wish to test them for being fair. Specifically, we would like to distinguish between the following cases:
\begin{itemize}
\item The machine picks the balls uniformly at random, that is, for any subset $B' \subseteq B$ of balls dropped into the machine, and for each $i \in B'$, the probability that the machine picks $i$ is $1/|B'|$;
\item The distribution according to which the balls are picked is $\epsilon$-far from uniform (where $\epsilon>0$ is some fixed constant, and the distance we consider is the standard variation distance defined above).
\end{itemize}
Suppose furthermore that we wish to distinguish between those cases as quickly as possible, and in particular, within few activations of the machine. Compare the following solutions.

We can use the uniformity tester \cite{uni1} for this task. Obtaining each sample from the underlying distribution requires one activation of the machine (with the entire set $B$), and we can complete the test using $\widetilde{\Theta}(\sqrt{|B|})$ activations.

Alternatively, using the algorithm we present in Section \ref{sec:uni}, using conditional samples we can complete the test using $O(1)$ activations only (the number of activations only has a polynomial dependency on $\epsilon$ and is logarithmic in the confidence parameter). Assuming that the drawing probabilities depend only on the physical characteristics of every ball separately, a conditional sample here corresponds to activating the machine with a specific subset of the balls rather than the entire set $B$.

This is for testing uniformity. Using the algorithm from Section \ref{sec:label_invar}, we could also test for any label-invariant property with $\poly(\log |B|)$ activations, which would allow us for example to give an estimation of the actual distance of the distribution from being uniform.

\subsubsection*{Asymmetric communication scenarios}
Suppose that two computers $A$ and $B$ are linked with an asymmetric communication link, in which transmitting information in one of the directions (say from $A$ to $B$) is much easier than in the other direction (consider e.g. a spacecraft traveling in remote space, with limited energy, computational power and transmitting capability; actually numerous examples of asymmetric communications also exist here on earth). Now assume that $B$ has access to some large data that can be modeled as collection of samples coming from an unknown distribution $\mu$,
while $A$ wants to learn or test some properties of $\mu$. We could simulate the standard testing algorithms by sending a request to $B$ whenever a random sample from $\mu$ is needed. Assuming that the most important measure of efficiency is how much information is sent by $B$, it would translate to the sample complexity of the simulated algorithm.

However, if $B$ can also produce conditional samples (for example if it has nearly unlimited cost-free access to samples from the distribution), then any property that is significantly easier to test with conditional samples can be tested with fewer resources here.

\subsubsection*{Political polls}
We mention these here because the modern-day practice of polling actually uses conditional sampling. Rather than taking a random sample of all willing potential participants, the polling population is usually first divided to groups according to common traits, and then each such group is polled separately before the results are re-integrated into the final prediction.

\subsection{Informal description of results}
In all sample-complexity upper bounds listed below there is a hidden factor of $\log(\delta^{-1})$, where $\delta$ is the maximal failure probability of the tester. Also, all lower bounds are for a fixed (and not very small) $\epsilon$. The results are summarized in Tables \ref{ubtable} and \ref{lbtable}.

\subsubsection*{Conditioned upon sets}

Testing algorithms in the conditional sampling model may be categorized according to the types of sets they condition upon. This is in addition to the questions of adaptivity and query complexity. The simplest types of sets would be constant sized sets. Another simple type of sets arises when we can endow the probability space with some linear order over the elements, and then only condition on sets which are intervals in this linear order. Actually, all of our adaptive testing algorithms use one of these types of sets. On the other hand, the non-adaptive algorithms seem to require the full generality of the model. This distinction was also made in \cite{CRS} (see below about this related work). The different types of sets used for conditional sampling are also summarized in Table \ref{ubtable}.

\begin{table}[b]
\center
\begin{tabular}{|l|c|c|}
\hline
{\bf Upper bounds} & \multicolumn{2}{|c|}{\bf Adaptive}  \\
\hline
 & Sample complexity & Conditioned sets \\
\hline
Uniformity & $\poly(\epsilon^{-1})$ & Constant size \\
\hline
Identity to known dist. & $\poly(\log^\star n,\epsilon^{-1})$ & $\poly(\log^\star n,\epsilon^{-1})$ size \\
\hline
Label-invariant prop. & $\poly(\log n,\epsilon^{-1})$ & Dyadic intervals \\
\hline
 & \multicolumn{2}{|c|}{\bf Non-adaptive} \\ 
\hline
 & Sample complexity & Conditioned sets \\
\hline
Uniformity & $\poly(\log n,\epsilon^{-1})$ & General\\
\hline
Identity to known dist. & $\poly(\log n,\epsilon^{-1})$ & General \\
\hline
\end{tabular}
\caption{Summary of upper bounds.}\label{ubtable}
\end{table}

\begin{table}[b]
\center
\begin{tabular}{|l|c|c|}
\hline
{\bf Lower bounds} & {\bf Adaptive} & {\bf Non-adaptive}\\
\hline
Uniformity and identity & --- & $\Omega(\log \log n)$ \\
\hline
Any label-invariant prop. & $\Omega(\sqrt{\log \log n})$ & (follows uniformity) \\
\hline
General properties & $\Omega(n)$ & (follows adaptive) \\
\hline
\end{tabular}
\caption{Summary of lower bounds.}\label{lbtable}
\end{table}

\subsubsection*{Adaptive testing}
The first result we prove is that uniformity, and more generally identity to any distribution that is very close to uniform in the $\ell_\infty$ norm, can be tested (adaptively) with $\poly(\epsilon^{-1})$ conditional samples (Theorem \ref{thm:uni} and Theorem \ref{thm:nearuni}, respectively). This is done by capturing (for far distributions) both ``light'' and ``heavy'' elements in the same small set and then conditioning over it. Our next result is that identity to any known distribution can be tested adaptively with $\poly(\log^\star n, \epsilon^{-1})$ conditional samples, where $n$ is the size of the domain (Theorem \ref{thm:ident}). This uses the uniformity result with the bucketing technique of \cite{Identity} together with a recursive argument.

A core result is that adaptive conditional samples are enough to construct an {\em explicit persistent sampler}. Such a sampler is essentially a way to simulate (unconditional) samples from a distribution $\tilde{\mu}$ that is close to $\mu$, and for which we can also provide exact probability queries like the oracle of \cite{ent1}.

From the construction of the explicit persistent sampler we derive our most general result that any label-invariant (i.e. invariant under permutation of the domain) property of distributions can be tested adaptively with $\poly(\log n, \epsilon^{-1})$ conditional samples (Theorem \ref{thm:invtest}). In fact, we go further to prove the following stronger result: with $\poly(\log n, \epsilon^{-1},\log(\delta^{-1}))$ conditional samples taken from $\mu$, it is possible to compute a distribution $\mu'$ that is $\epsilon$-close to $\mu$ up to some permutation of the domain $[n]$ (Theorem \ref{thm:learndist}). 

\subsubsection*{Non-adaptive testing}
We prove that uniformity can be tested non-adaptively with $\poly(\log n, \epsilon^{-1})$ conditional samples. Here too the tester enjoys a certain degree of tolerance, in the sense that it is possible to test identity with any distribution that is close enough to uniform (see Theorems \ref{thm:nonadap_uni} and \ref{thm:nanearuni}). This is done by first proving (through bucketing) that a portion of the ``total difference'' of $\mu$ from being uniform is in relatively equal-probability members of $[n]$, and then trying to capture just a few of them in a random set of an appropriate size. We also prove (from the uniformity test through standard bucketing arguments) that identity to any known distribution can be tested non-adaptively with $\poly(\log n, \epsilon^{-1})$ conditional samples (Theorem \ref{thm:nonadap_ident}).

\subsubsection*{Lower bounds}
As already mentioned in the introduction, adaptivity is useful when we have access to conditional sampling. We demonstrate this by proving that testing uniformity non-adaptively requires $\Omega(\log \log n)$ conditional samples, for some fixed $\epsilon>0$ (Theorem \ref{thm:lb_uni}). We also prove that the tester for any label-invariant property (from our main result) cannot be improved to work with a constant number of conditional samples: There is a label invariant property which requires $\Omega(\sqrt{\log \log n})$ samples to test, whether adaptively or not (Theorem \ref{thm:lb_label_invar}). Our third lower bound shows that for some properties conditional samples do not help much: There are distribution properties that cannot be tested (even adaptively) with $o(n)$ conditional samples (Theorem \ref{thm:genlb}). The first two lower bounds are through a special adaptation of Yao's method, while the last one is through a reduction to general properties of Boolean strings, of which maximally untestable examples are known.

\subsection*{About the gaps in the bounds}
We believe that for non-adaptive uniformity testing the upper bound is closer to the truth, in that the actual complexity should be close to logarithmic in $n$. A more careful analysis of the lower bound construction would be a good starting point towards narrowing the gap. We also believe that the correct lower bound for adaptive testing of general label-invariant properties is higher than our achieved one. Additionally we believe that an examination of the methods of \cite{valiant} should allow us to construct label-invariant properties for which testing in the traditional (unconditioned) sampling model is nearly useless.

\subsection{Related work}

Independently, Cannone et. al. \cite{CRS,CRSSoda} formulated the distribution testing with conditional samples model as well. In their work, Cannone et. al. achieve several results. Some of their results overlap with those of the present paper, but most of their work takes a different direction and emphasis.

\paragraph*{Uniformity}

Cannone et. al. give an algorithm for testing uniformity using $\tilde{O}(\epsilon^{-2})$ samples, and also give a lower bound of $\Omega(\epsilon^{-2})$. It is interesting to note that their upper bound only uses conditioning on sets of size $2$.

\paragraph*{Identity to a known distribution}

For this problem, Cannone et. al. demonstrate that conditioning on arbitrary sets is stronger than conditioning on sets of size $2$. They give an upper bound of $\tilde{O}(\epsilon^{-4}\log^4 n)$ and a lower bound of $\Omega\left(\sqrt{\frac{\log n}{\log\log n}}\right)$ for testing identity to a known distribution using samples conditioned on sets of size $2$, and an upper bound of $\tilde{O}(\epsilon^{-4})$ for testing it using samples on arbitrary sets.

\paragraph*{Identity between two unknown distributions}

The case of testing identity between two unknown distributions is an especially interesting one, as it showcases what seems to be a profound characteristic of the conditional sampling model. In developing an algorithm for this problem, Cannone et. al. introduce the notion of an ``approximate EVAL oracle''. Such an oracle is given some element $i\in [n]$ and should return a good estimate of the probability of $i$, while allowed to fail for some small fixed subset of $[n]$. This notion is somewhat reminiscent of the notion of an explicit persistent sampler used in Section \ref{sec:label_invar}.\footnote{A major difference is that our explicit persistent sampler will with high probability conform to exactly one distribution $\tilde{\mu}$ that is close to $\mu$, rather than just give approximate probability values for $\mu$.} Using this construction they give an algorithm that uses $\tilde{O}(\epsilon^{-4}\log^5 n)$ conditional samples to test identity between two unknown distributions. They also give an algorithm that uses $\tilde{O}(\epsilon^{-21}\log^6 n)$ samples, but only conditions on sets of size $2$.

\paragraph*{Estimating the distance from uniformity}

Cannone et. al. give an algorithm using $\tilde{O}(\epsilon^{-20})$ samples conditioned on sets of size $2$ to give an additive approximation for the distance of an unknown distribution from uniformity. It is interesting to note that this is a case of a label invariant property, though it outperforms the general algorithm for testing label invariant properties given in the present paper.

\paragraph*{Conditional samples over structured domains}

An interesting problem world arises when one tries to impose some sort of structure on the sets conditioned upon. The simplest case is when limiting their size, but one can imagine other cases. Cannone et. al. consider the case where the universe of elements is linearly ordered, and one may obtain samples conditioned on intervals. For this setting, they give an upper bound of $\tilde{O}(\epsilon^{-3}\log^3 n)$ samples and a lower bound of $\Omega\left(\frac{\log n}{\log\log n}\right)$ samples for testing uniformity. It is interesting to note that our explicit persistent sampler construction is also based only on samples conditioned on intervals (in fact, on the restricted set of dyadic intervals).



\section{Preliminaries}\label{sec:prelim}
\subsection{The conditional distribution testing model}
Let $\mu$ be a distribution over $\{1,\ldots,n\}$, its
probabilities denoted by $p_1,\ldots,p_n$, where $p_i=\Pr_\mu[i]$. We will also write $\mu(i)$ for $\Pr_\mu[i]$ where we deal with more then one distribution. The distribution $\mu$ is not known to the algorithm explicitly, and may only be accessed by drawing samples. A conditional distribution testing algorithm may submit any set $A\subseteq\{1,\ldots,n\}$ and receive a sample $i\in A$ that is drawn according to $\mu$ conditioned on $A$ (and independent
of any previous samples).

Thus when a sample is drawn according to $\mu$ conditioned on $A$, the probability of getting
$j$ is $\Pr[j|A]=p_j/(\sum_{i\in A} p_i)$ for $j \in A$ and $0$ for $j\not\in A$. If $\sum_{i\in A} p_i = 0$ then we assume (somewhat arbitrarily) that the algorithm obtains a uniformly drawn member of $A$.\footnote{See the beginning of Section \ref{sec:invlb} for how to essentially reduce a model without this assumption to this model.}

We measure {\em farness} using the variation distance: We say that $\mu$ is {\em $\epsilon$-far} from a property $\PP$ of distributions over $\{1,\ldots,n\}$, if for every $\mu'$ that satisfies $\PP$ and is described by $p'_1,\ldots,p'_n$ we have $d(\mu,\mu')=\frac12\sum_{i=1}^n |p_i-p'_i|\geq\epsilon$.

We will consider two types of conditional distribution testing algorithms. Non-adaptive testers, which must decide the conditioned sets to sample from before getting any samples, and adaptive testers, which have no such restriction.

\begin{defin}[Non-adaptive tester]
A {\em non-adaptive distribution tester for a property $\PP$ with conditional sample complexity $t:\reals \times \reals \times \naturals \to \naturals$} is a randomized algorithm, that receives $\epsilon,\delta >0$, $n \in \naturals$ and a conditional sampling oracle to a distribution $\mu$ over $[n]$, and operates as follows.
\begin{enumerate}
\item The algorithm generates a sequence of $t\le t(\epsilon,\delta,n)$ sets $A_1,\ldots,A_t \subseteq [n]$ (possibly with repetitions).
\item Then it calls the conditional oracle $t$ times with $A_1,\ldots,A_t$ respectively, and receives $j_1,\ldots,j_t$, where every $j_i$ is drawn according to the distribution $\mu$ conditioned on $A_i$, independently of $j_1,\ldots,j_{i-1}$ and any other history.
\item Based on the received elements $j_1,\ldots,j_t$ and its internal coin tosses, the algorithm accepts or rejects the distribution $\mu$.
\end{enumerate}
If $\mu$ satisfies $\PP$ then the algorithm must accept with probability at least $1-\delta$, and if $\mu$ is $\epsilon$-far from $\PP$ then the algorithm must reject with probability at least $1-\delta$.
\end{defin}

\begin{defin}[Adaptive tester]
An {\em adaptive distribution tester for a property $\PP$ with conditional sample complexity $t:\reals \times \reals \times \naturals \to \naturals$} is a randomized algorithm that receives $\epsilon,\delta >0$, $n \in \naturals$ and a conditional sampling oracle to a distribution $\mu$ over $[n]$ and operates as follows.
\begin{enumerate}
\item For $i\in\{1,\ldots,t\}$, where $t=t(\epsilon, \delta, n)$, at the $i$th phase the algorithm generates a set $A_i\subseteq [n]$, based on $j_1,\ldots,j_{i-1}$ and its internal coin tosses, and calls the conditional oracle with $A_i$ to receive an element $j_i$, drawn according to the distribution $\mu$ conditioned on $A_i$, independently of $j_1,\ldots,j_{i-1}$ and any other history.
\item Based on the received elements $j_1,\ldots,j_t$ and its internal coin tosses, the algorithm accepts or rejects the distribution $\mu$.
\end{enumerate}
If $\mu$ satisfies $\PP$ then the algorithm must accept with probability at least $1-\delta$, and if $\mu$ is $\epsilon$-far from $\PP$ then the algorithm must reject with probability at least $1-\delta$.
\end{defin}

As is standard in the field of property testing, the primary measure of efficiency of these testers is their sample complexity $t(\epsilon,\delta,n)$.

\subsection{Tools from previous works}
Our algorithms will make use of the Identity Tester of Batu et. al. \cite{Identity} (though it is important to note that this result is used mainly as a ``primitive'' and can be replaced in the sequel with just making enough samples to fully approximate the distribution).
\begin{theorem}[Identity Tester] \label{thm:batutester}
There is an algorithm $T$ for testing identity between an unknown distribution $\mu'$ and a known distribution $\mu$, both over $[n]$, with (ordinary) sample complexity $\tilde{O}(\sqrt{n}\mathrm{poly}(\epsilon^{-1})\log(\delta^{-1}))$. Namely, $T$ accepts with probability $1-\delta$ if $\mu'=\mu$ and rejects with probability $1-\delta$ if $\mu'$ is $\epsilon$-far from $\mu$.
\end{theorem}

We will also use the following inequality, which appears as Theorem A.1.11 and Theorem A.1.13 in \cite{AlonSpencer}:
\begin{lemma}\label{lem:atomcon}
Let $p_1,\ldots,p_n\in[0,1]$, $X_1,\ldots,X_n$ be fully independent random variables with $\Pr[X_i=1-p_i]=p_i$ and $\Pr[X_i=-p_i]=1-p_i$, and let $p=\frac1n\sum_{i=1}^n p_i$ and $X=\sum_{i=1}^n X_i$. Then $\Pr[|X|>a]<2\exp(-a^2/2pn)$.
\end{lemma}
When using this lemma we interpret $X+pn=\sum_{i=1}^n(X_i+p_i)$ as the number of successes in $n$ independent trials where the probability of success in the $i$th trial is $p_i$.

\subsubsection*{Bucketing} \label{sec:bucketing}
\newcommand{\bucketnum}[2]{\frac{\log {#2}}{\log(1 + {#1})}} 

Bucketing is a general tool, introduced in \cite{BatuCloseness, Identity},
that decomposes any explicitly given distribution into a collection
of distributions that are almost uniform. In this section we recall
the bucketing technique and lemmas from \cite{BatuCloseness, Identity}
that we will need for our proofs.

\begin{defin} Given a distribution $\mu$ over $[n]$, and
$M \subseteq [n]$ such that $\mu(M) > 0$, the \emph{restriction}
$\restrict{\mu}{M}$ is the distribution over $M$ with $\restrict{\mu}{M}(i) =
\mu(i)/\mu(M)$ (this is the the same as the conditioning of $\mu$ on $B$, only here we also change the domain).

Given a partition $\mathcal{M} = \{M_0, M_1, \dots, M_k\}$ of $[n]$, we denote by $\bucketdist{\mu}{\mathcal{M}}$ the distribution
 over $\{0\}\cup[k]$ in which $\bucketdist{\mu}{\mathcal{M}}(i) = \mu(M_i)$. This is the \emph{coarsening} of $\mu$ according to $\mathcal{M}$.
\end{defin}

\begin{defin}\label{def:bucketing}
Given an explicit distribution $\mu$ over $[n]$, $Bucket(\mu, [n],
\epsilon)$ is a procedure that generates a partition $\{M_0, M_1, \dots, M_k\}$ of the domain $[n]$, where $k
= \bucketnum{\epsilon}{n}<\frac2\epsilon\log(n)$. This partition satisfies the following conditions:
\begin{itemize}
 \item $M_0 = \{j \in [n] \mid \mu(j) < \frac{1}{n}\}$;
 \item for all $i \in [k]$, $M_i = \left\{ j \in [n] \mid \frac{(1+\epsilon)^{i-1}}{n} \leq \mu(j)
< \frac{(1+\epsilon)^{i}}{n}\right\}$.
\end{itemize}
\end{defin}

\begin{lemma}[Lemma 8 in \cite{Identity}]\label{lem:bucket} Let $\mu$ be a distribution over $[n]$ and let
$\{M_0, M_1, \dots, M_k\} \gets Bucket(\mu, [n], \epsilon)$. Then for all $i\in [k]$, $\norm{\restrict{\mu}{M_i} - \restrict{U}{M_i}}_\infty \leq \epsilon/n$.

\end{lemma}

\begin{lemma}[Lemma 6 in \cite{Identity}]\label{lem:partition} Let $\mu, \mu'$ be two distributions over $[n]$ and let the sequence of sets
$\mathcal{M} = \{M_0, M_1, \dots, M_k\}$ be a partition of $[n]$. If  $\norm{\restrict{\mu}{M_i} - \restrict{\mu'}{M_i}}_1 \leq \epsilon_1$ for every $i \in [k]$
and $\norm{\bucketdist{\mu}{\mathcal{M}} - \bucketdist{\mu}{\mathcal{M}}}_1 \leq \epsilon_2$, then $\norm{\mu - \mu'}_1 \leq \epsilon_1 + \epsilon_2$. Furthermore, $\norm{\mu - \mu'}_1 \leq \sum_{0\leq i \leq k}\mu(M_i)\norm{\restrict{\mu}{M_i}-\restrict{\mu'}{M_i}} + \epsilon_2$.
\end{lemma}

We reproduce the proof to obtain the ``furthermore'' claim:

\begin{proof}
This results from the following.
$$\norm{\mu-\mu'}_1 = \sum_{0\leq i \leq k}\sum_{j\in M_i} |\mu(j)-\mu'(j)| = \sum_{0\leq i \leq k}\sum_{j\in M_i} |\mu(M_i)\cdot\restrict{\mu}{M_i}(j)-\mu'(M_i)\cdot\restrict{\mu'}{M_i}(j)|$$
$$\leq \sum_{0\leq i \leq k}\sum_{j\in M_i} |\mu(M_i)\cdot\restrict{\mu}{M_i}(j)-\mu(M_i)\cdot\restrict{\mu'}{M_i}(j)|$$
$$\qquad+ \sum_{0\leq i \leq k}\sum_{j\in M_i} |\mu(M_i)\cdot\restrict{\mu'}{M_i}(j)-\mu'(M_i)\cdot\restrict{\mu'}{M_i}(j)|$$
$$=\sum_{0\leq i \leq k}\sum_{j\in M_i} \mu(M_i)|\cdot\restrict{\mu}{M_i}(j)-\restrict{\mu'}{M_i}(j)|+ \sum_{0\leq i \leq k}\sum_{j\in M_i} \restrict{\mu'}{M_i}(j)\cdot|\mu(M_i)-\mu'(M_i)|$$
$$=\sum_{0\leq i \leq k}\mu(M_i)\sum_{j\in M_i} \norm{\restrict{\mu}{M_i}(j)-\restrict{\mu'}{M_i}(j)}_1+ \sum_{0\leq i \leq k}|\mu(M_i)-\mu'(M_i)|$$
$$\leq \sum_{0\leq i \leq k}\mu(M_i)\sum_{j\in M_i} \norm{\restrict{\mu}{M_i}(j)-\restrict{\mu'}{M_i}(j)}_1+ \epsilon_2$$
This provides the ``furthermore'' claim. To obtain from the above the original claim note that $\sum_{0\leq i \leq k}\mu(M_i)\sum_{j\in M_i} \norm{\restrict{\mu}{M_i}(j)-\restrict{\mu'}{M_i}(j)}_1\leq\sum_{0\leq i \leq k}\mu(M_i)\epsilon_1=\epsilon_1$.
\end{proof}

\section{Adaptive testing for uniformity and identity}

In the following we formulate our testing algorithms to have a polynomial dependence on $\log(\delta^{-1})$. To make it linear in $\log(\delta^{-1})$ we can first run the algorithm $100\log(\delta^{-1})$ times with a fixed $\frac13$ error bound and then take the majority vote.

\subsection{Testing uniformity} \label{sec:uni}

\newcommand{\UnifQueries}[2]{(6/{#1})\log({#2}^{-1})} 
\newcommand{\BFClose}[2]{\frac{{#1}^2}{600\log({#2}^{-1})}} 
\newcommand{\BFError}[1]{{#1}/3}

\begin{theorem}\label{thm:uni}
There is an (adaptive) algorithm testing uniformity using $\mathrm{poly}(\epsilon^{-1}, \log(\delta^{-1}))$ conditional samples independently of $n$.
\end{theorem}

In fact we will prove something slightly stronger, which will be useful in the next sections:

\begin{theorem}[Near Uniformity Tester]\label{thm:nearuni} Let $\mu$ be a known distribution over $[n]$ such that $\norm{\mu - U_n}_{\infty} < \frac{\epsilon}{100n}$. Identity with $\mu$ can be tested using only $\mathrm{poly}(\epsilon^{-1}, \log(\delta^{-1}))$
conditional samples by an adaptive algorithm.
\end{theorem}

\begin{proof}
This follows from Algorithm \ref{alg:nearuni} by Lemmas \ref{lem:nearuniqc}, \ref{lem:nearunicomp} and \ref{lem:nearunisnd} below.
\end{proof}

Let $\mu'$ be the unknown distribution that is to be
sampled from.

\begin{algo}\label{alg:nearuni}(Near Uniformity Tester) The algorithm receives $\mu, \epsilon, \delta$ and $n$ and operates as follows.
\begin{enumerate}
\item Take $S$ to be $k=\UnifQueries{\epsilon}{\delta}$ independent samples according to $\mu'$ (unconditioned).
\item Take $U$ to be $k$ members of $\{1,\ldots,n\}$ chosen uniformly at random.
\item Invoke the Identity Tester of Theorem \ref{thm:batutester} to check whether $\mu'\upharpoonright_{U\cup
    S}$ is $\BFClose{\epsilon}{\delta}$-close to $\mu\upharpoonright_{U\cup S}$ over $U\cup S$ with bounded error
  probability $\BFError{\delta}$, and answer as the tester did.
\end{enumerate}
\end{algo}

\begin{lemma}\label{lem:nearuniqc}
The sample complexity of Algorithm \ref{alg:nearuni} is $\mathrm{poly}(\epsilon^{-1}, \log(\delta^{-1}))$.
\end{lemma}

\begin{proof}
The algorithm draws $k$ samples, and then invokes the closeness tester on a set of size $2k$ and an error parameter polynomial in $\epsilon^{-1}$. Since the sample complexity of the closeness tester is polynomial in the support size and error parameter, and $k=\UnifQueries{\epsilon}{\delta}$, the total sample complexity of Algorithm \ref{alg:nearuni} is $\mathrm{poly}(\epsilon^{-1}, \log(\delta^{-1}))$.
\end{proof}

\begin{lemma}\label{lem:nearunicomp}
If $d(\mu,\mu')=0$ then Algorithm \ref{alg:nearuni} accepts with probability at least $1-\delta$.
\end{lemma}

\begin{proof}
 If $\norm{\mu - \mu'}_1 = 0$ then
 $\norm{\mu\upharpoonright_{U\cup S} -
   \mu'\upharpoonright_{U\cup S}}_1 = 0$ and then the algorithm
 will accept if the closeness tester does, which will happen with probability at least $1-\frac{\delta}{3}$.
\end{proof}

Let the individual probabilities for the distribution $\mu$ be
denoted by $p_1,\ldots,p_n$ and the probabilities for the distribution
$\mu'$ denoted by $p_1',\ldots,p'_n$.
We first note that

$$2d(\mu,\mu')=\norm{\mu - \mu'}_1 =
\sum_{i=1}^n|p_i-p_i'|=2\sum_{p_i'< p_i}(p_i - p_i') = 2\sum_{p_i'>p_i}(p_i'-p_i)$$

Assume from now on that this distance is at least $2\epsilon$ (which corresponds to variation distance at least $\epsilon$).

\begin{lemma}\label{lem:biggies}
With probability at least $1-\delta/3$ we have an $i\in S$ for which
$(p_i' - p_i) \geq \frac{\epsilon}{2n}$.
\end{lemma}

\begin{proof}
Clearly $\sum_{p_i<p_i'<p_i+\epsilon/2n}(p_i'-p_i)<\frac12\epsilon$. Therefore:
$$\sum_{p_i'\geq p_i + \epsilon/2n}p_i'>\sum_{p_i'\geq p_i +
  \epsilon/2n}(p_i'-p_i)=\sum_{p_i'>p_i}(p_i'-p_i)-\sum_{p_i<p_i'<
  p_i + \epsilon/2n}(p_i'-p_i)>\frac12\epsilon$$

This means that after $\UnifQueries{\epsilon}{\delta}$ samples, with probability at least $1-\delta/3$ we will get an $i$ with such a $p_i'$ into $S$.
\end{proof}

\begin{lemma}\label{lem:smallies}
With probability at least $1-\delta/3$ we have an $i\in U$ for which $p_i'<p_i$.
\end{lemma}

\begin{proof}
Note that $\sum_{p_i'<p_i}(p_i-p_i')\leq
|\{i:p_i'<p_i\}| \cdot \max\{p_i\}$. Now since $\max_i\{p_i\} < (1 +
  \frac{\epsilon}{100})\frac1n$ there are at least $(\epsilon/2) n$ such
$i$. A uniformly random choice of $\UnifQueries{\epsilon}{\delta}$ indexes will get one
with probability at least $1-\delta/3$.
\end{proof}

\begin{lemma}\label{lem:condsampfar}
When both events above occur, $\restrict{\mu'}{U\cup S}$ is at least $\BFClose{\epsilon}{\delta}$-far from $\restrict{\mu}{U\cup S}$ over $U\cup S$.
\end{lemma}

\begin{proof}
Note that $|S\cup U|=2k=2\cdot \UnifQueries{\epsilon}{\delta}$, and that the two events above
mean that there are $i$ and $j$ in this set such that $p'_i\geq\frac{1+\epsilon/2}{1+ \epsilon/100}p'_j$.
Denoting the conditional probabilities $q_i=p_i/\mu(S\cup U)$ and $q'_i=p'_i/\mu'(S\cup U)$, we note that we obtain $q'_i\geq\frac{1+\epsilon/2}{1+ \epsilon/100}q'_j$, while both $q_i$ and $q_j$ are bounded between $\frac{1-\epsilon/100}{1+\epsilon/100}\frac1{2k}$ and $\frac{1+\epsilon/100}{1-\epsilon/100}\frac1{2k}$.
Therefore, either $q'_i>q_i+\frac{\epsilon}{40k}$ or $q'_j<q_j-\frac{\epsilon}{40k}$. Either way, $d(\restrict{\mu}{U\cup S},\restrict{\mu'}{U\cup S})>\frac{\epsilon}{100k}$, which concludes the proof.
\end{proof}

This concludes the soundness proof, as the last step of the algorithm
checks the closeness of $\restrict{\mu'}{U\cup S}$ to $\restrict{\mu}{U\cup S}$ with this
approximation parameter. Thus we obtain:

\begin{lemma}\label{lem:nearunisnd}
Let $\mu$ be a known distribution over $[n]$.
  Then if  $\norm{\mu - U_n}_{\infty} < \frac{\epsilon}{100n}$ and $d(\mu,\mu')>\epsilon$ then Algorithm \ref{alg:nearuni} rejects with probability at least $1-\delta$.
\end{lemma}

\begin{proof}
Follows from a union bound for the events of Lemma \ref{lem:biggies} and Lemma \ref{lem:smallies}, and the failure probability of the test invoked in the last step of the algorithm (due to Lemma \ref{lem:condsampfar}).
\end{proof}

\subsection{Testing identity to a known distribution}

\newcommand{\IdentityQueries}[3]{{4{#1}^{-1}\log^\star({#3})\log({#2}^{-1})}} %
\newcommand{\IdentityBucketing}[1]{\frac{{#1}}{200\log^\star m}} %
\newcommand{\IdentityReductionErr}[2]{\frac{{#1}{#2}}{12\log^\star (m)\log({#1}^{-1})}} %
\newcommand{\IdentityReductionDist}[2]{\frac{{#1}}{2\log^\star {#2}}} %
\newcommand{\IdentityRecursionErr}[1]{\frac{{#1}}{3}} %
\newcommand{\IdentityRecursionDist}[2]{{#1}\left(1 - \frac{1}{\log^\star {#2}}\right)} %

Recall that if we define $\log^{(0)}(n)=n$ and by induction $\log^{(k+1)}(n)=\log(\log^{(k)}(n))$, then the $\log^\star$ function is defined by $\log^\star(n)=\min\{k:\log^{(k)}(n)\leq 1\}$.

\begin{theorem} \label{thm:ident}
 Testing identity with a known distribution can be
  done by an adaptive algorithm using $\mathrm{poly}(\log^\star n, \epsilon^{-1}, \log(\delta^{-1}))$ conditional samples.
\end{theorem}

\begin{proof}
This follows from Algorithm \ref{alg:identity} by Lemmas \ref{lem:identityqc}, \ref{lem:identitycomp} and \ref{lem:identitysnd} below.
\end{proof}

Let $\mu$ be the known distribution and $\mu'$ be the
unknown distribution that is accessed by sampling.
The following is an algorithm for testing identity to the known
distribution $\mu$ over $[n]$. In the initial run we feed it $m=n$, but in the recursive runs it keeps track of $m$ as the ``original $n$''.

\begin{algo}\label{alg:identity} (Identity Test) The algorithm receives $\epsilon$, $\delta$, $n$, $m$ and $\mu$, operating as follows.
\begin{enumerate} \setcounter{enumi}{-1}
\item If $n\leq\left(\frac{400\log(1/\epsilon)}{\epsilon}\log^{\star} m\right)^3$ then perform a brute-force test: Take ${100\log(1/\delta)}\epsilon^{-2}n^2\log n$ samples to write a distribution $\tilde{\mu}$ that is $\frac\epsilon2$-close to $\mu'$ (with probability $1-\delta$); if $d(\tilde{\mu},\mu)\leq\frac\epsilon2$ then ACCEPT and otherwise REJECT.
\item Let $\mathcal{M}=\{M_0, M_1, \dots, M_k\} \gets Bucket(\mu, [n], \IdentityBucketing{\epsilon})$.
\item Sample $r = \IdentityQueries{\epsilon}{\delta}{m}$ elements from $\mu'$. Let
  $M_{i_1}, \dots, M_{i_{r}}$ be the buckets where these elements lie.
\item For every bucket $M_{i_1}, \dots, M_{i_{r}}$ test using the
 Near Uniformity Test (Theorem \ref{thm:nearuni}) whether
  $\norm{\restrict{\mu}{M_{i_j}} - \restrict{\mu'}{M_{i_j}}}_1 \geq \IdentityReductionDist{\epsilon}{m}$ with error bound $\IdentityReductionErr{\delta}{\epsilon}$.\label{line:identityreduction}
\item If for any $i_j$ we have $\norm{\restrict{\mu}{M_{i_j}} - \restrict{\mu'}{M_{i_j}}}_1
  \geq \IdentityReductionDist{\epsilon}{m}$ then REJECT.
\item Else recursively test if $\norm{\bucketdist{\mu}{\mathcal{M}} - \bucketdist{\mu'}{\mathcal{M}}}_1 \leq \IdentityRecursionDist{\epsilon}{m}$ with error bound $\IdentityRecursionErr{\delta}$. If not then REJECT else ACCEPT.
\end{enumerate}
\end{algo}

First, we bound the number of recursion levels that can occur.

\begin{lemma}\label{lem:idshallow}
Algorithm \ref{alg:identity} never enters more than $2\log^\star(n)$ recursion levels from the initial $n=m$ call.
\end{lemma}

\begin{proof}
Note that in the first $2\log^\star(n)$ recursion levels, the distance parameter that is passed is still at least $\IdentityRecursionDist{\epsilon}{n}^{2\log^\star(n)}>\frac\epsilon{e^2}$, so we will prove the bound on the number of levels even if this is the distance parameter that is used in all but the first level.
If $\log(n)\leq\left(\frac{400\log(1/\epsilon)}{\epsilon}\log^{\star} m\right)$ then after at most one recursion level the test goes to the brute force procedure in Step 0 and ends. Otherwise, note that the recursive call now receives $n'\leq\frac{400e^2\log(n)\log^\star(m)}{\epsilon}\leq\log^3(n)$, and that call itself will make a recursive call with universe size $n''\leq  \frac{1200e^2\log\log(n)\log^\star(m)}{\epsilon}\leq \log n$ (unless it already terminated for some other reason). This is sufficient for the bound.
\end{proof}

\begin{lemma}\label{lem:identitycomp}
If $d(\mu,\mu')=0$ then Algorithm \ref{alg:identity} accepts with probability at least $1-\delta$.
\end{lemma}

\begin{proof}
The base case where $n\leq\left(\frac{400\log(1/\epsilon)}{\epsilon}\log^{\star} m\right)^3$ is clear.
Otherwise, if $\norm{\mu - \mu'}_1 = 0$ then for all
buckets $M_i$ we have  $\norm{\restrict{\mu}{M_i} - \restrict{\mu'}{M_i}}_1 = 0$ and
 $\norm{\bucketdist{\mu}{\mathcal{M}} - \bucketdist{\mu'}{\mathcal{M}}}_1 = 0$. From Lemma \ref{lem:bucket} we know that
$\norm{\restrict{\mu}{M_i} - \restrict{U}{M_i}}_{\infty} \leq \IdentityBucketing{\epsilon}\cdot\frac{1}{n}\leq\frac{\epsilon'}{100n}$, where $\epsilon'$ is the distance parameter fed to the Near Uniformity Tester, and
hence the Near Uniformity tester (Theorem \ref{thm:nearuni}) is applicable and will accept with probability $1-\IdentityReductionErr{\delta}{\epsilon}$. Taking the union bound over the number of samples taken and the probability of failure for the recursive call (recall that a recursive call adds a $\frac13$ factor to $\delta$) gives us the desired bound.
\end{proof}

For soundness we need the following lemma.

\begin{lemma}\label{lem:genbucket} If $\norm{\mu - \mu'}_1 \geq \epsilon$
  then for any $t$  at least one of the following two will happen:
\begin{enumerate}
\item $\sum_{\{i: \norm{\restrict{\mu}{M_i} - \restrict{\mu'}{M_i}}_1
    \geq \epsilon/2t\}} \mu(M_i)  \geq \epsilon/2t$
\item $\norm{\bucketdist{\mu}{\mathcal{M}} - \bucketdist{\mu'}{\mathcal{M}}}_1 \geq \epsilon(1 - 1/t)$
\end{enumerate}
\end{lemma}

\begin{proof}  Recall Lemma \ref{lem:partition}:
$$\norm{\mu - \mu'}_1 \leq  \sum_{0\leq i \leq k}\mu(M_i)\cdot \norm{\restrict{\mu}{M_i}-\restrict{\mu'}{M_i}}_1 + \norm{\bucketdist{\mu}{\mathcal{M}} - \bucketdist{\mu'}{\mathcal{M}}}_1$$

Thus if $\norm{\bucketdist{\mu}{\mathcal{M}} - \bucketdist{\mu'}{\mathcal{M}}}_1 < \epsilon(1-1/t)$ and
$\sum_{\{i: \norm{\restrict{\mu}{M_i} - \restrict{\mu'}{M_i}}_1
    \geq \epsilon/2t\}} \mu(M_i)  < \epsilon/2t$
then we have
  $\norm{\mu - \mu'}_1 < \epsilon$, a contradiction.
\end{proof}

\begin{lemma}\label{lem:identitysnd}
  If $d(\mu,\mu')>\epsilon$ then Algorithm \ref{alg:nearuni} rejects with probability at least $1-\delta$.
\end{lemma}

\begin{proof}
The base case of $n\leq\left(\frac{400\log(1/\epsilon)}{\epsilon}\log^{\star} m\right)^3$ is clear.
Refer now to Lemma \ref{lem:genbucket}, taking $t=\log^\star m$. Assume that we are in the first case of the lemma, that is $\sum_{\{i: \norm{\restrict{\mu}{M_i} - \restrict{\mu'}{M_i}}_1
    \geq \epsilon/2t\}} \mu(M_i)  \geq \epsilon/2t$. therefore, the probability of sampling an index for which the test in Line \ref{line:identityreduction} should reject is at least $\frac{\epsilon}{2\log^\star m}$. This implies that the probability that one of the sampled elements is such is at least $\delta/3$, and since the probability that all calls to the Near Uniformity Test fail is at most $\delta/3$ as well, we accept with probability at most $2\delta/3$.

Now assuming that we are in the second case of Lemma \ref{lem:genbucket}, by the induction hypothesis we reject with probability at least $\delta/3$. Thus the overall error probability is at most $\delta$.
\end{proof}

\begin{lemma}\label{lem:identityqc}
The sample complexity of Algorithm \ref{alg:identity} is $\mathrm{poly}(\log^\star n, \epsilon^{-1}, \log(\delta^{-1}))$.
\end{lemma}

\begin{proof}
If $n\leq\left(\frac{400\log(1/\epsilon)}{\epsilon}\log^{\star} m\right)^3$ then it is polynomial in $\epsilon$ and $\log^{\star} m$, and so is the result of substituting it in the number of queries of the brute force check of Step 0, $q_b(\epsilon,\delta,n)={100\log(1/\delta)}\epsilon^{-2}n^2\log n$.
For analyzing the sample complexity when the above does not hold for $m=n$, let $q(\epsilon, \delta, n)$ denote the sample complexity
of the algorithm. By the algorithm's definition, we have the following formula, where $q_{u}$ is the sample complexity of the Near Uniformity Tester:
$$q(\epsilon, \delta, n, m)\leq\IdentityQueries{\epsilon}{\delta}{m}\left(1+q_{u}\left(\IdentityReductionDist{\epsilon}{m}, \IdentityReductionErr{\delta}{\epsilon},n\right)\right)$$
$$+q\left(\IdentityRecursionDist{\epsilon}{m},\IdentityRecursionErr{\delta},\frac{400\log(n)\log^\star(m)}{\epsilon},m\right)$$

According to Lemma \ref{lem:idshallow}, after at most $2\log^\star n$ recursion levels from the initial $n=m$, the right hand side is now within the realm of the brute force check, and we get a summand bounded by $q_b(\epsilon/e^2,\delta\cdot3^{-2\log^\star n}, \left(\frac{400\log(1/\epsilon)}{\epsilon}\log^{\star} n\right)^3)=\poly(\log^\star n,\epsilon^{-1},\log(\delta^{-1}))$. Therefore:

$$q(\epsilon, \delta, n, n)\leq 8\epsilon^{-1}(\log^\star n)^2\log(\delta^{-1})\left(1+q_{u}\left(\frac{\epsilon}{2e^2\log^\star n}, \frac{\epsilon\cdot\delta\cdot3^{-2\log^\star n}}{40e^2(\log^\star n)^2\log(\delta^{-1})},n\right)\right)$$
$$ + \poly(\log^\star n,\epsilon^{-1},\log(\delta^{-1}))$$

Since by Lemma \ref{lem:nearuniqc}, the Near Uniformity Tester has sample complexity polynomial in the distance parameter and polylogarithmic in the error bound, we obtain the statement of the lemma.
\end{proof}

\section{Non-adaptive testing for uniformity and identity}

Recall that a non-adaptive tester must be able to produce all the conditioned upon sets in advance. In this section we show that these weaker testers can still beat testers without conditional sampling.

\subsection{Testing uniformity}

\begin{theorem}\label{thm:nonadap_uni}
Testing uniformity can be done
  using $\mathrm{poly}(\log n, \epsilon^{-1}, \log(\delta^{-1}))$ non-adaptive conditional samples.
\end{theorem}

Again, we will actually prove the following stronger statement:

\begin{theorem}[Nonadaptive Near Uniformity Tester] \label{thm:nanearuni} Let $\mu$ be a known distribution over $[n]$.
  If  $\norm{\mu - U_n}_{\infty} < \epsilon/8n$ then
  identity with $\mu$ can be tested using $\mathrm{poly}(\log n,\epsilon^{-1},\log(\delta^{-1}))$
conditional samples by a non-adaptive algorithm.
\end{theorem}

\begin{proof}
For $\delta=1/3$, this follows from Algorithm \ref{alg:nanearuni} by Lemmas \ref{lem:nanearuniqc}, \ref{lem:nanearunicomp} and \ref{lem:nanearunisnd} below. For general $\delta$ we use a standard amplification technique: We repeat the algorithm $\Theta(\log(\delta^{-1}))$ times (with independent probabilities) and take the majority vote. This obviously incurs a multiplicative factor of $\Theta(\log(\delta^{-1}))$ in the sample complexity.
\end{proof}

\newcommand{\BigSetSample}[2]{64{#1}^{-2}\log^2({#2})}
\newcommand{\SmallestBigSet}[2]{\lceil\log(2000{#1}^{-6}\log^5({#2}))\rceil}
\newcommand{\BiggestBigSet}[1]{\lceil\log({#1})\rceil}
\newcommand{\SmallSet}[2]{9000{#1}^{-6}\log^5({#2})}
\newcommand{\SmallSetEst}[2]{{#1}{#2}}


\begin{algo}\label{alg:nanearuni}
The algorithm is given $n,\epsilon$ and $\mu$, and has nonadaptive conditional sample access to $\mu'$.
\begin{enumerate}
\item For $\SmallestBigSet{\epsilon}{n}\leq j\leq\BiggestBigSet{n}$, set $U_j$ to be a uniformly random set of $\min\{n,2^j\}$ indices.
\item For every $U_j$, perform $\BigSetSample{\epsilon}{n}$ conditional samples, and if the same index was drawn twice, REJECT. \label{line:bigsetsample}
\item Uniformly pick a random set $U$ of $\SmallSet{\epsilon}{n}$ elements, and invoke the Identity Tester of Theorem \ref{thm:batutester} to test whether $\restrict{\mu'}{U}=\restrict{\mu}{U}$ or $d(\restrict{\mu'}{U},\restrict{\mu}{U})>\frac{\epsilon}{24|U|}$ with success probability $\frac{19}{20}$. \label{line:smallsetsample}
\item ACCEPT unless any of the above testers rejected.
\end{enumerate}
\end{algo}

\begin{lemma}\label{lem:nanearunicomp}
If $d(\mu,\mu')=0$ then Algorithm \ref{alg:nearuni} accepts with probability at least $2/3$.
\end{lemma}

\begin{proof}
Since $\norm{\mu - U_n}_{\infty} < \epsilon/8n$, the probability that an element will be drawn twice in the $j$th iteration of Line \ref{line:bigsetsample} is at most $\binom{\BigSetSample{\epsilon}{n}}{2}\cdot\left(\frac{1+\epsilon/8}{1-\epsilon/8}\right)^2\cdot2^{-2j}$. Summation over all values of $j$ gives us less than $1/9$.

Since $\mu=\mu'$, $\restrict{\mu'}{U}=\restrict{\mu}{U}$ for any $U\subseteq [n]$, and the probability that Line \ref{line:smallsetsample} rejects is at most $1/9$. This obtains the error bound in the lemma.
\end{proof}

The following is immediate from the algorithm statement and Theorem \ref{thm:batutester}:

\begin{lemma}\label{lem:nanearuniqc}
The sample complexity of Algorithm \ref{alg:identity} is $\mathrm{poly}(\log n,\epsilon^{-1})$.
\end{lemma}

\begin{proof}
This follows from the number of samples used in Lines \ref{line:bigsetsample} and \ref{line:smallsetsample} and the fact that Line \ref{line:bigsetsample} is iterated at most $\log n$ times.
\end{proof}

In the following we assume that $d(\mu,\mu')>\epsilon$.

Let $M_1,M_2,\ldots,M_k$ be the bucketing of $\mu$ and $M'_1,M'_2,\ldots,M'_k$ the bucketing of $\mu'$, both with $\epsilon/3$. Denote the individual probabilities by $p_1,\ldots,p_n$ and $p_1',\ldots,p_n'$ respectively.

\begin{lemma}\label{lem:bigbuck}
$|M'_0\cup M'_1|\geq \epsilon n$ and there exists $2<j\leq k$ such that $|M'_j|\geq\frac{\epsilon^2n}{96(1+\epsilon/3)^j\log n}$.
\end{lemma}

\begin{proof}
Note that $[n]=M_0\cup M_1$ by our requirement from $\mu$. Now following Lemma \ref{lem:smallies}, $\sum_{p_i'<p_i}(p_i-p_i')\leq
|\{i:p_i'<p_i\}| \cdot \max\{p_i\}$. Now since $\max_i\{p_i\} < (1 +
  \epsilon/8)\frac1n$ there are at least $(\epsilon/2) n$ such $i$.

For the second part we will adapt the proof of Lemma \ref{lem:biggies}. Clearly $\sum_{p_i<p_i'<p_i+11\epsilon/12n}(p_i'-p_i)<\frac{11}{12}\epsilon$. Therefore:
$$\sum_{p_i'\geq p_i + 11\epsilon/12n}p_i'>\sum_{p_i'\geq p_i +
  11\epsilon/12n}(p_i'-p_i)=\sum_{p_i'>p_i}(p_i'-p_i)-\sum_{p_i<p_i'<
  p_i + 11\epsilon/12n}(p_i'-p_i)>\frac{1}{12}\epsilon$$

Since $p_i\geq \frac{1-\epsilon/8}{n}$, we know that the $p'_i$ in the left hand side have (assuming $\epsilon<1/10$)
$$p_i'\geq \frac{1-\epsilon/8}{n}+\frac{11\epsilon}{12n}=\frac{1+19\epsilon/24}{n} \geq \frac{(1+\epsilon/3)^{2}}{n}$$
and therefore all these $p_i'$s are in buckets $M'_j$ for $2<j\leq k$.

Since $k=\frac{\log n}{\log(1+\epsilon/3)}$, there exists some $2<j\leq k$ such that $\mu'(M'_j)\geq\frac{\epsilon \log(1+\epsilon/3)}{12\log n}$. By the definition of the buckets this gives $|M'_j|\geq \frac{\epsilon \log(1+\epsilon/3)}{12\log n}\cdot\frac{n}{(1+\epsilon/3)^j}>\frac{\epsilon^2n}{96(1+\epsilon/3)^j\log n}$.
\end{proof}

\begin{lemma}\label{lem:fetcha}
Given a set $B$ of size $l$, a set $U$ of $\min\{n,\frac{3n}{l}\}$ indices chosen uniformly at random will with probability more than $\frac{19}{20}$ contain a member of $B$.
\end{lemma}

\begin{proof}
The probability is lower bounded by the probability for $3n/l$ indexes chosen uniformly and independently with repetitions from $[n]$ to intersect $B$, which is $1-(1-l/n)^{\frac{3n}{l}}\geq \frac{19}{20}$.
\end{proof}

\begin{lemma}\label{lem:nanearunisnd}
Let $\mu$ be a known distribution over $[n]$.
  If  $\norm{\mu - U_n}_{\infty} < \epsilon/8n$ and $d(\mu,\mu')>\epsilon$ then Algorithm \ref{alg:nearuni} rejects with probability at least $2/3$.
\end{lemma}

\begin{proof}
We partition into cases according to the $j$ guaranteed by Lemma \ref{lem:bigbuck}.

If $(1+\frac{\epsilon}{3})^{j}\leq 40\epsilon^{-4}\log^4 n$, then $|M'_j|\geq\frac{\epsilon^6}{3000\log^5 n }n$, so by Lemma \ref{lem:fetcha} with probability $\frac{19}{20}$ the set $U$ in Line \ref{line:smallsetsample} will contain a member $h$ of $M'_j$. Note that $j>2$ and therefore $\mu'(h)\geq\frac{(1+\epsilon/3)^2}{n}$. By the first part of Lemma \ref{lem:bigbuck} with probability $\frac{19}{20}$ (actually much more than that) we will also sample an element $l\in M'_0\cup M'_1$. Thus we have $\mu'(h)\geq(1+\epsilon/3)\mu'(l)$, and also $\restrict{\mu'}{U}(h)\geq(1+\epsilon/3)\restrict{\mu'}{U}(l)$, while both $\restrict{\mu}{U}(h)$ and $\restrict{\mu}{U}(l)$ are restricted between $\frac{1-\epsilon/8}{1+\epsilon/8}\frac{1}{|U|}$ and $\frac{1+\epsilon/8}{1-\epsilon/8}\frac{1}{|U|}$. Therefore, either $\restrict{\mu'}{U}(h)>\restrict{\mu}{U}(h) +\frac{\epsilon}{12|U|}$ or $\restrict{\mu'}{U}(l)<\restrict{\mu}{U}(l) -\frac{\epsilon}{12|U|}$. Either way $d(\restrict{\mu'}{U},\restrict{\mu}{U})>\frac{\epsilon}{24|U|}$, which will be identified by the tester of Theorem \ref{thm:batutester} with probability $\frac{19}{20}$. Thus in total we get a rejection probability greater than $\frac79$.

Otherwise, let $i$ be such that the value $2^i$ is between $\min\{n,300\epsilon^{-2}\log n(1+\frac\epsilon3)^{j}\}$ and $2\min\{n,300\epsilon^{-2}\log n(1+\frac\epsilon3)^{j}\}$ (recall the lower bound on $(1+\frac\epsilon3)^{j}$). In that case the $U_i$ in Line \ref{line:bigsetsample} will with probability at least $\frac{19}{20}$ contain a member $a$ of $M'_j$. Additionally, the expected value of $\mu'(U_i)$ is $\min\{1,\frac{2^i}{n}\}\leq\min\{1,\frac{600}{n}\epsilon^{-2}(1+\frac\epsilon3)^{j}\log n\}$, thus by Markov's inequality, with probability at least $\frac89$ we will have $\mu'(U_i)\leq \min\{1,\frac{6000}{n}\epsilon^{-2}(1+\frac\epsilon3)^{j}\log n\}$. Therefore, $\restrict{\mu'}{U_i}(a)\geq \frac{\epsilon^2}{6000(1+\epsilon/3)\log n}$. Thus the expected number of times $a$ is sampled is at least $\frac{\log n}{125}$ and therefore by Lemma \ref{lem:atomcon} with probability $1-2\exp(-\frac{\log n}{250})$ we will sample $a$ at least twice. Thus in total we get a rejection probability greater than $\frac79$ for $n>2^{253}$ (this lower bound can be reduced for the price of a higher degree polynomial dependence on $\log n$).
\end{proof}

\subsection{Testing identity to a known distribution}

\begin{theorem} \label{thm:nonadap_ident}
Identity to a known distribution can be tested
  using $\mathrm{poly}(\log n, \epsilon^{-1}, \log(\delta^{-1}))$ non-adaptive conditional samples.
\end{theorem}

\begin{proof}
This follows from Algorithm \ref{alg:naidentity} by Lemmas \ref{lem:naidentqc}, \ref{lem:naidentcomp} and \ref{lem:naidentsnd} below.
\end{proof}

Let $\mu$ be  the known distribution and $\mu'$ be the
unknown distribution that is accessed by sampling.
The following is an algorithm for testing identity with the known
distribution $\mu$ over $[n]$:

\begin{algo}\label{alg:naidentity}(Identity Test) The algorithm receives $\epsilon$, $\delta$, $n$ and $\mu$
  and operates as follows.
\begin{enumerate}
\item Let $\mathcal{M}=\{M_0, M_1, \dots, M_k\} \gets Bucket(\mu, [n], \frac{\epsilon}{8})$.
\item For each bucket $M_{1}, \dots, M_{k}$ test using the
 Nonadaptive Near Uniformity Test (Theorem \ref{thm:nanearuni}) to check whether
  $\norm{\restrict{\mu}{M_{j}} - \restrict{\mu'}{M_{j}}}_1 \geq \epsilon/2$ with error bound $\frac{\delta\log(1+\epsilon/8)}{2\log n}$, rejecting immediatly if any test rejects.
\item Invoke the Identity Tester of Theorem \ref{thm:batutester} to test if $\norm{\bucketdist{\mu}{\mathcal{M}} - \bucketdist{\mu'}{\mathcal{M}}}_1 \leq \epsilon/2$ with error bound $\delta/2$, answering as the test does.
\end{enumerate}
\end{algo}

\begin{lemma}\label{lem:naidentcomp}
If $d(\mu,\mu')=0$ then Algorithm \ref{alg:naidentity} accepts with probability at least $1-\delta$.
\end{lemma}

\begin{proof}
In this case, for all buckets $\norm{\restrict{\mu}{M_{j}} - \restrict{\mu'}{M_{j}}}_1=0$ and $\norm{\bucketdist{\mu}{\mathcal{M}} - \bucketdist{\mu'}{\mathcal{M}}}_1 = 0$, and thus by the union bound we obtain the statement.
\end{proof}

\begin{lemma}\label{lem:naidentqc}
The sample complexity of Algorithm \ref{alg:identity} is $\mathrm{poly}(\log n, \epsilon^{-1}, \log(\delta^{-1}))$.
\end{lemma}

\begin{proof}
We invoke the Nonadaptive Near Uniformity Test $\frac{\log n}{\log(1+\epsilon/8)}$ times, and invoke the Closeness Tester with a distribution of support size $\frac{\log n}{\log(1+\epsilon/8)}$. Therefore by Lemma \ref{lem:nanearuniqc} and Theorem \ref{thm:batutester} we obtain the bound in the statement.
\end{proof}

\begin{lemma}\label{lem:naidentsnd}
If $d(\mu,\mu')>\epsilon$, then Algorithm \ref{alg:naidentity} rejects with probability at least $1-\delta$.
\end{lemma}

\begin{proof}
Assume that the test accepted. If no error was made, then by Lemma \ref{lem:partition} we have that $d(\mu,\mu')\leq\epsilon$. By the union bound the probability of error is at most $\delta$.
\end{proof}

\section{Explicit persistent samplers}\label{sec:exp_samp}

We exhibit here the strength of the conditional sampling oracle, using it to implement explicit persistent samplers as defined below. 

\begin{defin}
Given a distribution over distributions $\mathcal{M}$, a {\em $(\delta,s)$-explicit persistent sampler} is an algorithm that can be run up to $s$ times (and during each run may store information to be used in subsequent runs), that in every run returns a pair $(i,\eta)$. It must satisfy that with probability at least $1-\delta$, the $i$'s for all $s$ runs are independent samples of a single distribution $\tilde{\mu}$ that in itself was drawn according to the distribution over distributions $\mathcal{M}$, and every output pair $(i,\eta)$ satisfies $\eta=\tilde{\mu}(i)$.
\end{defin}

The goal of this section is to construct, for every distribution $\mu$, an explicit persistent sampler for a distribution over distributions that are all close to $\mu$, which uses a conditional sampling oracle for $\mu$.

Note that although the definition does not require it, the explicit persistent samplers we construct will also be able to answer oracle queries of the form ``what is the probability of $i$?''.

In all the following we assume that $n$ is a power of $2$, as otherwise we can ``pad'' the probability space with additional zero-probability members.

\subsection{Ratio trees and reconstituted distributions}
\newcommand{\threshold}[2]{\frac{#1}{2\log({#2})}}
\newcommand{\finerat}[2]{(\frac{#1}{2\log({#2})})^2}

The main driving force in our algorithm for constructing an explicit sampler is a way to estimate the ratio between the distribution weight of two disjoint sets. To make it into a weight oracle for a value $i\in [n]$, we will use successive partitions of $[n]$, through a fixed binary tree. Remember that here $n$ is assumed to be a power of $2$.

We first define how to ``reconstruct'' a distribution from a tree with ratios, and afterward show how to put the ratios there.

\begin{defin}\label{def:recons}
Let $T$ be a (full) balanced binary tree with $n$ leaves labeled by $[n]$. Let $U$ be the set of non-leaf nodes of the tree, and assume that we have a function $\alpha:U\to [0,1]$. For $u\in U$ denote by $L(u)$ the set of leaves that are descendants of the left child of $u$, and by $R(u)$ the leaves that are descendants of the right child of $u$.

The {\em reconstituted distribution} according to $\alpha$ is the distribution $\tilde{\mu}$ that is calculated for every $i\in [n]$ as follows:
\begin{itemize}
\item Let $u_1,\ldots,u_{\log(n)+1}$ be the root to leaf path for $i$ (so in particular $u_{\log(n)+1}=i$).
\item For ever $1\leq j\leq\log n$, set $p_j=\alpha(u_j)$ if $i$ is a descendant of the left child of $u_j$ (that is if $i\in L(u_j)$), and otherwise set $p_j=1-\alpha(u_j)$.
\item Set $\tilde{\mu}(i)=\prod_{j=1}^{\log n}p_j$.
\end{itemize}
\end{defin}

For intuition, note the following trivial observation.

\begin{obs}
If for a distribution $\mu$ we set $\alpha(u)=\frac{\mu(L(u))}{\mu(L(u))+\mu(R(u))}$, using an arbitrary value (say $\frac12$) for the case where $\mu(L(u))+\mu(R(u))=0$, then the reconstituted distribution $\tilde{\mu}$ is identical to $\mu$.
\end{obs}

However, if we only have conditional oracle access to $\mu$ then we cannot know the values $\frac{\mu(L(u))}{\mu(L(u))+\mu(R(u))}$. The best we can do the the following.

\begin{defin}\label{def:ratioest}
An \emph{$(\epsilon, \delta)$-ratio estimator} for $T$ and a distribution $\mu$ is an algorithm $A$ that given a non-leaf vertex $u\in U$ outputs a number $r$, such that with probability $1-\delta$ we have that $\frac{\mu(L(v))}{\mu(L(v))+\mu(R(v))}-\epsilon\leq r\leq \frac{\mu(L(v))}{\mu(L(v))+\mu(R(v))}+\epsilon$.
\end{defin}

\begin{algo}\label{alg:ratioest}(Ratio Estimator)
The algorithm is given a balanced binary tree $T$ with $n$ leaves, a non-leaf vertex $u\in U$ and parameters $\epsilon, \delta$. It also has conditional sample access to a distribution $\mu$.
\begin{enumerate}
\item Sample $t=2\epsilon^{-2}\log(\delta^{-1})$ elements according to $\restrict{\mu}{L(u)\cup R(u)}$, and let $s$ be the number of samples that are in $L(u)$.
\item Return the ratio $\frac{s}{t}$ of the samples that are in $L(u)$ to the total number of samples.
\end{enumerate}
\end{algo}

\begin{lemma}\label{lem:estratio}
For any $\epsilon, \delta$ Algorithm \ref{alg:ratioest} is an $(\epsilon, \delta)$-ratio estimator for $T$ and $\mu$ which uses $t=2\epsilon^{-2}\log(\delta^{-1})$ non-adaptive conditional samples from $\mu$.
\end{lemma}

\begin{proof}
The number of samples used is immediate. Let us now proceed to show that this is indeed an $(\epsilon, \delta)$-ratio estimator. The expected value of $\frac{s}{t}$ is $\frac{\mu(L(u))}{\mu(L(u))+\mu(R(u))}$.

By Chernoff's inequality, the probability that $\frac{s}{t}$ deviates from its expected value by an additive term of more than $\epsilon$ is at most $2\exp(-2\epsilon^2 \cdot t)$. By our choice of $t$ we obtain the statement.
\end{proof}

If we could ``populate'' the entire tree $T$ (through the function $\alpha$) by values that do not deviate by much from the corresponding ratios, then we would be able to create an estimate for $\mu$ that is good for most values.

\begin{defin}
The function $\alpha:U\to [0,1]$ is called {\em $\epsilon$-fine} if $|\alpha(u)-\frac{\mu(L(u))}{\mu(L(u))+\mu(R(u))}|\leq\finerat{\epsilon}{n}$ for every $u\in U$.

We call a distribution $\tilde{\mu}$ {\em $\epsilon$-fine} if there exists a set $B$ such that $\mu(B)\leq\epsilon$, and additionally $\tilde{\mu}(i)=(1\pm\epsilon)\mu(i)$ for every $i\in [n]\setminus B$.
\end{defin}

\begin{lemma}\label{lem:finerat}
If $\alpha$ is $\epsilon$-fine then the reconstituted distribution $\tilde{\mu}$ is $\epsilon$-fine.
\end{lemma}

\begin{proof}
To define the set $B$, for every $i$ consider the $p_1,\ldots,p_{\log n}$ that are set as per Definition \ref{def:recons}, and set $i\in B$ if and only if there exist some $p_j$ that is smaller than $\threshold{\epsilon}{n}$. Next, denote by $q_1,\ldots,q_k$ the ``intended'' values, that is $q_j=\frac{\mu(L(u_j))}{\mu(L(u_j))+\mu(R(u_j))}$ if $i\in L(u_j)$ and $q_j=\frac{\mu(R(u_j))}{\mu(L(u_j))+\mu(R(u_j))}$ otherwise. Noting that $p_j$ does not deviates from $q_j$ by more than $\finerat{\epsilon}{n}$, an induction over $\log n$ (the height of $T$) gives that $1-\mu(B)$ is at least $(1-\frac{\epsilon}{\log n})^{\log n}>1-\epsilon$.

For $i\in [n]\setminus B$, we note that in this case $p_j=(1\pm\frac{\epsilon}{2\log n})q_j$, and hence $\tilde{\mu}(i)=\prod_{j=1}^{\log n}p_j=(1\pm\frac{\epsilon}{2\log n})^{\log n}\prod_{j=1}^{\log n}q_j=(1\pm\epsilon)\mu(i)$.
\end{proof}

We should note here that it is not hard to prove that an $\epsilon$-fine distribution $\tilde{\mu}$ is of distance not more than $4\epsilon$ from the original $\mu$. However, we will in fact refer to yet another distribution which will be easier to estimate, so we will show closeness to it instead.

\begin{defin}\label{def:trim}
Given an $\epsilon$-fine distribution $\tilde{\mu}$ and its respective set $B$, its {\em $\epsilon$-trimmed distribution} $\overline{\mu}$ is a distribution over $[n]\cup\{0\}$ defined by the following.
\begin{itemize}
\item For $i\in B\cup\{i:\tilde{\mu}(i)<\frac{\epsilon}{n}\}$ we set $\overline{\mu}(i)=0$. For such $i$ we also set $j_i=0$.
\item For all other $i\in [n]$ we set $j_i$ to be the largest integer for which $\frac{(1+\epsilon)^{j_i-1}}{n}\epsilon\leq \tilde{\mu}(i)$, and set $\overline{\mu}(i)=\frac{(1+\epsilon)^{j_i-1}}{n}\epsilon$.
\item Finally set $\overline{\mu}(0)=1-\sum_{i=1}^n\overline{\mu}(i)$; note that $\overline{\mu}(i)\leq\tilde{\mu}(i)$ for all $1\leq i\leq n$ and hence $\overline{\mu}(0)\geq 0$.
\end{itemize}

The {\em $\epsilon$-renormalized distribution} $\hat{\mu}$ over $[n]$ is just the conditioning $\restrict{\overline{\mu}}{[n]}$.
\end{defin}

It will be important later to note that the renormalized distribution is in fact (a permutation of) the tentative distribution according to $m_0,\ldots,m_k$, where for $0\leq j\leq k$ we set $m_j=|\{i:j_i=j\}|$, as per Definition \ref{def:tentative} below.

\begin{lemma}\label{lem:renorm}
The renormalized distribution $\hat{\mu}$ corresponding to an $\epsilon$-fine distribution $\tilde{\mu}$ is $4\epsilon$-close to $\mu$.
\end{lemma}

\begin{proof}
First we consider the trimmed distribution $\overline{\mu}$, and its distance from $\mu$ (when we extend it by setting $\mu(0)=0$). Recalling that this variation distance is equal to $\sum_{\{i:\overline{\mu}(i)<\mu(i)\}}(\mu(i)-\overline{\mu}(i))$, we partition the set of relevant $i$'s into two subsets.
\begin{itemize}
\item For those $i$ that are in $B$ (for which $\overline{\mu}(i)=0$), the total difference is $\mu(B)\leq\epsilon$.
\item For any other $i$ for which $\overline{\mu}(i)<\mu(i)$, note that $\overline{\mu}(i)\geq\frac{1}{1+\epsilon}\tilde{\mu}(i)\geq\frac{1-\epsilon}{1+\epsilon}\mu(i)>(1-3\epsilon)\mu(i)$. This means that the sum over differences for all such $i$ is bounded by $3\epsilon$.
\item We never have $\overline{\mu}(0)<\mu(0)$.
\end{itemize}
Thus the distance between $\overline{\mu}$ and $\mu$ is not more than $4\epsilon$. As for $\hat{\mu}$, the sum of differences over $i$ for which $\hat{\mu}(i)<\mu(i)$ is only made smaller (the conditioning only increases the probability for every $i>0$), and so the $4\epsilon$ bound remains.
\end{proof}

\subsection{Distribution samplers and learning}

To construct an explicit sampler we need to not only sample from the distribution $\mu$, but to be able to ``report'' $\mu(i)$ for every $i$ thus sampled. This we cannot do, but it turns out that we can sample from a close distribution $\tilde{\mu}$ while reporting $\tilde{\mu}(i)$. In fact we will sample from a distribution that in itself will be drawn from the following distribution over distributions.

\begin{defin}
The \emph{$(\epsilon,\delta)$-condensation} of $\mu$ is the distribution over $\epsilon$-fine distributions (with respect to $\mu$) that is defined by the following process.
\begin{itemize}
\item Let $T$ be a (full) balanced binary tree whose leaves are labeled by $[n]$, and $U$ be its set of internal nodes.
\item For every $u\in U$, let $\alpha(u)$ be the (randomized) result of running the corresponding $(\finerat{\epsilon}{n},\delta)$-Ratio Estimator (Algorithm \ref{alg:ratioest}), when conditioned on this result indeed being of distance not more than $\finerat{\epsilon}{n}$ away from $\frac{\mu(L(u_j))}{\mu(L(u_j))+\mu(R(u_j))}$. This is done independently for every $u$.
\item The drawn distribution $\tilde{\mu}$ is the reconstituted distribution according to $T$ and $\alpha$
\end{itemize}
\end{defin}

The algorithm that we define next is an explicit persistent sampler: It is explicit in that it relays information about $\tilde{\mu}(i)$ along with $i$, and persistent in that it simulates (with high probability) a sequence of $s$ independent samples from the same $\tilde{\mu}$.

\begin{algo}\label{alg:sampler}(Persistent Sampler)
The algorithm is given parameters $\epsilon, \delta$ and $s$, and has conditional sample access to a distribution $\mu$.
\begin{enumerate}
\item On the initial run,  set $T$ to be a full balanced binary tree with $n$ leaves labeled by $[n]$. Let $w$ denote the root vertex and $U$ denote the set of non-leaf vertices. $\alpha$ is initially unset.
\item On all runs, set $u_1=w$, and repeat the following for $l=1,\ldots,\log n$.
 \begin{enumerate}
 \item If $\alpha(u_l)$ is not set yet, set it to the result of the $(\finerat{\epsilon}{n},\frac{\delta}{s\log n})$-Ratio Estimator (Algorithm \ref{alg:ratioest}); run it independently of prior runs.
 \item Independently of any prior choices, and without sampling from $\mu$, with probability $\alpha(u_l)$ set $u_{l+1}$ to be the left child of $u_l$ and $p_l=\alpha(u_l)$, and with probability $1-\alpha(u_l)$ set $u_{l+1}$ to be the right child of $u_l$ and $p_l=1-\alpha(u_l)$.

 \end{enumerate}
 \item Set $i$ to be the label of the leaf $u_{\log n+1}$ and $\eta=\prod_{l=1}^{\log n}p_l$. Return $i$ and $\eta$.
\end{enumerate}
\end{algo}

\begin{lemma}\label{lem:sampler}
For any $\epsilon, \delta$ and $s$, Algorithm \ref{alg:sampler} is a $(\delta, s)$-explicit persistent sampler for the $(\epsilon,\frac{\delta}{s\log n})$-condensation of $\mu$. It uses a total of $2^5\cdot\epsilon^{-4}\log^5n\cdot\log(s\delta^{-1}\log n)$ many adaptive conditional samples from $\mu$ to output a sample.
\end{lemma}

\begin{proof}
The calculation of the number of samples is straightforward (but note that these are adaptive now). During $s$ runs, by the union bound with probability at least $1-\delta$ all of the calls to the $(\finerat{\epsilon}{n},\frac{\delta}{s\log n})$-Ratio Estimator produced results that are not more than $(\finerat{\epsilon}{n}$-away from the actual rations.

Conditioned on the above event, the algorithm acts the same as the algorithm that first chooses for every $u\in U$ the value $\alpha(u)$ according to a run of the $(\finerat{\epsilon}{n},\frac{\delta}{s\log n})$-Ratio Estimator conditioned on it being successful, and only then traverses the tree $T$ for every required sample. The latter algorithm is identical to picking a distribution $\tilde{\mu}$ according to the $(\epsilon,\frac{\delta}{s\log n})$-condensation of $\mu$, and then (explicitly) sampling from it.
\end{proof}

\section{Testing any label-invariant property}\label{sec:label_invar}

\setcounter{theorem}{0}

We show here the following ``universal testing'' theorem for label-invariant properties.

\begin{theorem}\label{thm:invtest}
Every label-invariant property of distributions can be tested adaptively
using at most $\mathrm{poly}(\log n,\epsilon^{-1},\log(\delta^{-1}))$ conditional samples.
\end{theorem}

It is in fact a direct corollary of the following learning result.

\begin{theorem}\label{thm:learndist}
There exist an algorithm that uses $\mathrm{poly}(\log n,\epsilon^{-1},\log(\delta^{-1}))$ adaptive conditional samples to output a distribution $\tilde{\mu}$ over $[n]$, so that with probability at least $1-\delta$ some permutation of $\tilde{\mu}$ will be $\epsilon$-close to $\mu$.
\end{theorem}

\begin{proof}
The required algorithm is Algorithm \ref{alg:learndist} below, by Lemma \ref{lem:learndist}.
\end{proof}

To derive Theorem \ref{thm:invtest}, use Theorem \ref{thm:learndist} to obtain a distribution $\tilde{\mu}$ that is $\epsilon/2$-close to a permutation of $\mu$, and then accept $\mu$ if and only if $\tilde{\mu}$ is $\epsilon/2$-close to the tested property.

In a similar manner, one can also derive the following corollaries:

\begin{coro}
There exist an algorithm that uses $\mathrm{poly}(\log n,\epsilon^{-1},\log(\delta^{-1}))$ adaptive conditional samples to test whether two unknown distributions are identical up to relabeling.
\end{coro}

\begin{coro}
For every label-invariant property $P$, there exist an algorithm that uses $\mathrm{poly}(\log n,\epsilon^{-1},\log(\delta^{-1}))$ adaptive conditional samples, accepts any distribution $\epsilon/2$-close to $P$ with probability at least $1-\delta$ and rejects any distribution $\epsilon$-far from $P$ with probability at least $1-\delta$.
\end{coro}

The main idea of the proof of Theorem \ref{thm:learndist} is to use a bucketing, and try to approximate the number of members of every bucket, which allows us to construct an approximate distribution. However, there are some roadblocks, and in the foremost the fact that we cannot really query the value $\mu(i)$. Instead we will use an explicit persistent sampler as introduced in Section~\ref{sec:exp_samp}.

\subsection{Bucketing and approximations}
\newcommand{\epsbucketnum}[2]{\frac{\log {#2}\log({#1}^{-1})}{\log^2(1 + {#1})}} 

We need a bucketing that also goes into smaller probabilities than those needed for the other sections.

\begin{defin}
Given an explicit distribution $\mu$ over $[n]$, $Bucket'(\mu, [n],
\epsilon)$ is a procedure that generates a partition $\{M_0, M_1, \dots, M_k\}$ of the domain $[n]$, where $k
= \epsbucketnum{\epsilon}{n}$. This partition satisfies the following conditions:
\begin{itemize}
 \item $M_0 = \{j \in [n] \mid \mu(j) < \frac{\epsilon}{n}\}$;
 \item for all $i \in [k]$, $M_i = \left\{ j \in [n] \mid \frac{(1+\epsilon)^{i-1}}{n}\epsilon \leq \mu(j)
< \frac{(1+\epsilon)^{i}}{n}\epsilon\right\}$.
\end{itemize}
\end{defin}

In the rest of this section, bucketing will always refer to this version. Also, from here on we fix $\epsilon$ and $k=\epsbucketnum{\epsilon}{n}$ as above (as well as mostly ignore floor and ceiling signs). We also assume that $\epsilon$ is small enough, say smaller than $\frac1{100}$.

Suppose that we have $m_0,\ldots,m_k$, where $m_i=|M_i|$ is the size of the $i$'th set in the bucketing of a distribution $\mu$. Then we can use these to construct a distribution that is guaranteed to be close to some permutation of $\mu$.

\begin{defin}\label{def:tentative}
Given $m_0,\ldots,m_k$ for which $\sum_{j=0}^km_j=n$ and $\epsilon$, the {\em tentative distribution} over $[n]$ is the one constructed according to the following.
\begin{itemize}
\item Set $r_1,\ldots,r_n$ so that $|\{i:r_i=0\}|=m_0$ and $|\{i:r_i=\frac{(1+\epsilon)^{j-1}}{n}\epsilon\}|=m_j$ for every $1\leq j\leq k$ (the order of $r_1,\ldots,r_n$ is arbitrary).
\item Set a distribution $\tilde{\mu}$ over $[n]$ by setting $\mu(i)$ equal to $r_i/\sum_{j=1}^nr_i$.
\end{itemize}
\end{defin}

To gain some intuition, note the following.

\begin{obs}\label{obs:tentative}
If $M_0,\ldots,M_k$ is the bucketing of $\mu$ and $\tilde{\mu}$ is the tentative distribution according to $m_0=|M_0|,\ldots,m_k=|M_k|$, then $\tilde{\mu}$ is $2\epsilon$-close to some permutation of $\mu$.
\end{obs}

\begin{proof}
We assume that we have already permuted $\tilde{\mu}$ so that each $\tilde{\mu}(i)$ refers to an $r_i$ set according to the bucket $M_j$ satisfying $i\in M_j$ (such a permutation is possible because here we used the actual sizes of the buckets).

We recall that the distance is in particular equal to $\sum_{\{i:\tilde{\mu}(i)<\mu(i)\}}(\mu(i)-\tilde{\mu}(i))$. Referring to the $r_i$ of the definition above, we note that in this case $\sum_{i=0}^nr_i\leq \sum_{i=0}^n\mu(i)=1$ and hence $\tilde{\mu}(i)\geq r_i$. For $i\not\in M_0$, this means that $\tilde{\mu}(i)\geq (1-\epsilon)\mu(i)$. For the rest we just note that $\sum_{i\in M_0}\mu(i)\leq\epsilon$. Together we get the required bound.
\end{proof}

The above observation essentially states that it is enough to find the numbers $m_0,\ldots,m_k$ associated with $\mu$. However, the best we can hope for is to somehow estimate the size, or total probability, of every bucket. The following shows that this is in fact sufficient.

\begin{defin}
Given $\alpha_0,\ldots,\alpha_k$ for which $\sum_{j=0}^k\alpha_j=1$, the {\em bucketization} thereof is the sequence of integers $\hat{m_0},\ldots,\hat{m_k}$ defined by the following.
\begin{itemize}
\item For any $1\leq j\leq k$ let $\hat{m_j}$ be the integer closest to $n\alpha_k$ (where an ``exact half'' is arbitrarily rounded down).
\item If $\sum_{j=1}^k\hat{m_j}>n$, then decrease the $\hat{m_j}$ until they sum up to $n$, each time picking $j$ to be the smallest index for which $\hat{m_j}>0$ and decreasing that quantity by $1$.
\item Finally set $\hat{m_0}=n-\sum_{j=1}^k\hat{m_j}$.
\end{itemize}
We say that the bucketization has {\em failed} if in the second step we had to decrease any $\hat{m_j}$ for which $\frac{(1+\epsilon)^{j-1}}{n}\epsilon\geq\frac{\epsilon}{k}$.
\end{defin}

\begin{lemma}\label{lem:bucking}
Suppose that $m_0,\ldots,m_k$, $\alpha_0,\ldots,\alpha_k$ are such that :
\begin{itemize}
\item $\sum_{j=0}^km_j=n$
\item $\sum_{j=1}^km_j\frac{(1+\epsilon)^{j-1}}{n}\epsilon\leq 1$
\item $\sum_{j=0}^k\alpha_j=1$
\item $|m_j-\alpha_j|\frac{(1+\epsilon)^{j-1}}{n}\epsilon<\frac{\epsilon}{2k}$ for all $1\leq j\leq k$
\end{itemize}
  and let $\hat{m_0},\ldots,\hat{m_k}$ be the bucketization of $\alpha_0,\ldots,\alpha_k$. Then $\hat{m_0},\ldots,\hat{m_k}$ are all well defined (the bucketization process did not fail), and additionally if $\tilde{\mu}$ is the tentative distribution according to $m_0,\ldots,m_k$ and $\hat{\mu}$ is the tentative distribution according to $\hat{m_0},\ldots,\hat{m_k}$, then the distance between $\hat{\mu}$ and $\tilde{\mu}$ (after some permutation) is at most $4\epsilon$.
\end{lemma}

\begin{proof}
First thing to note is that $m_j=\hat{m_j}$ for all $j$ for which $\frac{(1+\epsilon)^{j-1}}{n}\epsilon\geq\frac{\epsilon}{k}$, before the decreasing step, so there will be no need to decrease these values and the bucketization will not fail.

For all $j\geq 1$, before decreasing some of the $\hat{m_j}$ we have that $|m_j-\hat{m_j}|\frac{(1+\epsilon)^{j-1}}{n}\epsilon<\frac{\epsilon}{k}$ (if $\frac{(1+\epsilon)^{j-1}}{n}\epsilon\leq\frac{\epsilon}{k}$ then the distance is not more than doubled by the rounding, and otherwise it follows from $m_j=\hat{m_j}$). Since the bucketization did not fail, the decreasing step only affects values $\hat{m_j}$ for which $\frac{(1+\epsilon)^{j-1}}{n}\epsilon<\frac{\epsilon}{k}$, and the total required decrease in them was by not more than $k$ (as the rounding in the first step of the bucketization added no more than $1$ to each of them), we obtain the total bound $\sum_{j=1}^k|m_j-\hat{m_j}|\frac{(1+\epsilon)^{j-1}}{n}\epsilon\leq3\epsilon$.

Let $r_i$ denote the corresponding values in the definition of $\tilde{\mu}$ being the tentative distribution according to $m_0,\ldots,m_k$, and $\hat{r_i}$ be the analog values in the definition of $\hat{\mu}$ being the tentative distribution according to $\hat{m_0},\ldots,\hat{m_k}$. By what we already know about $\sum_{j=1}^k|m_j-\hat{m_j}|\frac{(1+\epsilon)^{j-1}}{n}$ we have in particular $\sum_{i=1}^n\hat{r_i}=\sum_{i=1}^nr_i\pm3\epsilon$. Combined with the known bounds on $\sum_{i=1}^nr_i$, we can conclude by finding a permutation for which we can bound $\sum_{i=1}^n|r_i-\hat{r_i}|$ by $3\epsilon$, which will give the $4\epsilon$ bound on the distribution distance $\frac12\sum_{i=1}^n|\tilde{\mu}(i)-\hat{\mu}(i)|$.

The permutation we take is the one that maximizes the number of $i$'s for which $r_i=\hat{r_i}$; for the value $\frac{(1+\epsilon)^{j-1}}{n}\epsilon$ we can find $\min\{m_j,\hat{m_j}\}$ such $i$'s (for every $1\leq j\leq k$), and the hypothetical worst case is that whenever $r_i\neq\hat{r_i}$ one of them is zero (sometimes the realizable worst case is in fact not as bad as the hypothetical one). Thus the $\sum_{j=1}^k|m_j-\hat{m_j}|\frac{(1+\epsilon)^{j-1}}{n}\epsilon\leq3\epsilon$ bound leads to the $4\epsilon$ bound on the distribution distance.
\end{proof}

A problem still remains, in that sampling from $\mu$ will not obtain a value $\alpha_j$ close enough to the required $m_j\frac{(1+\epsilon)^{j-1}}{n}\epsilon$. The variations in the $\mu(i)$ inside the bucket $M_j$ itself could be higher than the $\frac{\epsilon}{2k}$ that we need here. In the next subsection we will construct not only a ``bucket identifying'' oracle, but tie it with an explicit persistent sampler that will simulate the approximate distribution rather than the original $\mu$.

\subsection{From bucketing to learning}

An explicit persistent sampler is almost sufficient to learn the distribution. The next step would be to estimate the size of a bucket of the $\epsilon$-fine distribution $\tilde{\mu}$ by explicit sampling (i.e.\ getting the samples along with their probabilities). However, Lemma \ref{lem:bucking} requires an approximation not of $\tilde{\mu}(M_j)$ (where $M_j$ is a bucket of $\tilde{\mu}$) but rather of $|M_j|\frac{(1+\epsilon)^{j-1}}{n}\epsilon$. In other words, we really need to approximate $\overline{\mu}(M_j)$, where $\overline{\mu}$ is the corresponding trimmed distribution.

Therefore we define the following explicit sampler for an $\epsilon$-trimmed distribution. We ``bend'' the definition a little, as this sampler will not be able to provide the corresponding probability for $i=0$.

\begin{algo}\label{alg:trsampler}(Trimming Sampler)
The algorithm is given parameters $\epsilon, \delta$ and $s$, and has conditional sample access to a distribution $\mu$.
\begin{enumerate}
\item Run the Persistent Sampler (Algorithm \ref{alg:sampler}) with parameters $\epsilon, \delta$ and $s$ to obtain $i$ and $\eta$; additionally retain $p_1,\ldots,p_{\log n}$ as calculated during the run of the Persistent Sampler.
\item If there exists $l$ for which $p_l<\threshold{\epsilon}{n}$ then return ``$0$''.
\item If $\eta<\frac{\epsilon}{n}$ then return ``$0$''.
\item Otherwise, let $j$ be the largest integer for which $\frac{(1+\epsilon)^{j-1}}{n}\epsilon\leq\eta$, and set $\eta'=\frac{(1+\epsilon)^{j-1}}{n}\epsilon$.
\item With probability $1-\eta'/\eta$ return ``$0$'', and with probability $\eta'/\eta$ return $(i,j)$ (where $j$ corresponds to $\overline{\mu}(i)=\eta'$).
\end{enumerate}
\end{algo}

The following observation is now easy.

\begin{obs}\label{obs:trsampler}
The trimming sampler (Algorithm \ref{alg:trsampler}) is a $(\delta,s)$-persistent sampler, and explicit whenever the returned sample is not $0$, for the distribution over distributions that results from taking the $\epsilon$-trimming of an $\epsilon$-fine distribution $\tilde{\mu}$ and its corresponding $B$ that was drawn according to the $(\epsilon,\frac{\delta}{s\log n})$-condensation of $\mu$. The algorithm uses in total $2^5\cdot\epsilon^{-4}\log^5n\cdot\log(s\delta^{-1}\log n)$ many adaptive conditional samples from $\mu$ to output a sample.
\end{obs}

\begin{proof}
The number of samples is inherited from Algorithm \ref{alg:sampler} as no other samples are taken. The algorithm switches the return value to ``$0$'' whenever $i\in B$ (as defined in the proof of Lemma \ref{lem:finerat}), and otherwise returns ``$0$'' exactly according to the corresponding conditional probability difference for $i$ between $\tilde{\mu}$ (as in the definition of a reconstituted distribution) and $\overline{\mu}$ (as in the definition of the corresponding trimmed distribution). Finally, whenever the returned sample is $i>0$ the algorithm clearly returns the corresponding $j_i$ (see Definition \ref{def:trim}).
\end{proof}

We are now ready to present the algorithm providing Theorem \ref{thm:learndist}.

\begin{algo}\label{alg:learndist}(Distribution Approximation)
The algorithm is given parameters $\epsilon, \delta$, and has conditional sample access to a distribution $\mu$.
\begin{enumerate}
\item Set $s=2^{12}\epsilon^{-4}\log^2(n)\log(\delta^{-1})$, and $k=\frac{\log n\log (12\epsilon^{-1})}{\log^2(1+\epsilon/12)}$ (the number of buckets in an $\epsilon/12$-bucketing of a distribution over $[n]$).
\item Take $s$ samples through the $(\epsilon/12,\delta/2,s)$-Trimming Sampler.
\item Denote by $s_0$ the number of times that the sampler returned ``$0$'', and for $1\leq j\leq k$ denote by $s_j$ the number of times that the sampler returned $(i,j)$ for any $i$.
\item Let $m'_0,\ldots,m'_k$ be the bucketization of $\alpha_0=\frac{s_0}{s},\ldots,\alpha_k=\frac{s_k}{s}$.
\item Return the tentative distribution according to $m'_0,\ldots,m'_k$.
\end{enumerate}
\end{algo}

\begin{lemma}\label{lem:learndist}
The Distribution Approximation algorithm (Algorithm \ref{alg:learndist}) will with probability at least $1-\delta$ return a distribution that is $\epsilon$-close to a permutation of $\mu$. This is performed using at most $\tilde{O}(\epsilon^{-8}\log^{7}n\log^2(\delta^{-1}))$ conditional samples.
\end{lemma}

\begin{proof}
The number of samples is immediate from the algorithm statement and Observation \ref{obs:trsampler}.

By Observation \ref{obs:trsampler}, with probability at least $1-\delta/2$ all samples of the Trimming Sampler will be from one $\epsilon/12$-trimming of some $\epsilon/12$-fine distribution $\overline{\mu}$. Set $m_0=|\{1\leq i\leq n:\overline{\mu}(i)=i\}|$ and for $1\leq j\leq k$ set $m_j=|\{i:\overline{\mu}(i)=\frac{(1+\epsilon)^{j-1}}{n}\epsilon\}|$. Recall that the $\epsilon/12$-renormalized distribution corresponding to $\overline{\mu}$ is in fact the tentative distribution according to $m_0,\ldots,m_k$. By Lemma \ref{lem:renorm}, this distribution is $\epsilon/2$-close to $\mu$.

Note now that for every $1\leq j\leq k$ the expectation of $\alpha_j$ is exactly $m_j\frac{(1+\epsilon/12)^{j-1}}{n}\epsilon/12$. By virtue of a Chernoff bound and the union bound, our choice of $s$ implies that with probability $1-\delta/2$ (conditioned on the previous event) we in fact get values that satisfy $|m_j-\alpha_j|\frac{(1+\epsilon/12)^{j-1}}{n}\epsilon/12<\frac{\epsilon/12}{2k}$ for every $1\leq j\leq k$. This satisfies the assertions of Lemma \ref{lem:bucking}, and thus the tentative distribution according to $m'_0,\ldots,m'_k$ will be $\epsilon/2$-close to the tentative distribution according to $m_0,\ldots,m_k$, and hence will be $\epsilon$-close to $\mu$.
\end{proof}

Note that if we were to use this algorithm for testing purposes, the dependence on $\delta^{-1}$ can be made logarithmic by setting it to $1/3$ and repeating the algorithm $\log(\delta^{-1})$ times, taking majority (but it may not be possible if we are interested in $\overline{\mu}$ itself).

\section{Lower bounds for label invariant properties}\label{sec:invlb}

In this section we prove two sample complexity lower bounds for testing label-invariant distribution properties in our model. The first is for testing uniformity, and applies to non-adaptive algorithms. The second bound is for testing whether a distribution is uniform over some subset $U \subseteq \{1,\ldots,n\}$ of size exatcly $2^{2k}$ for some $k$, and applies to general (adaptive) algorithms.

The analysis as it is written relies on the particular behavior of our model when conditioning on a set of probability zero, but this can be done away with: Instead of a distribution $\mu$ with probabilities $p_1,\ldots,p_n$ over $[n]$, we can replace it with the $o(1)$-close distribution $\hat{\mu}$ with probabilities $\hat{p_1},\ldots,\hat{p_i}$ where $\hat{p_i}=\frac1{n^2}+(1-\frac1n)p_i$. The same analysis of why an algorithm will fail to correctly respond to $\mu$ will pass on to $\hat{\mu}$, which has no zero probability sets.

\subsection{Preliminary definitions}

We start with some definitions that are common to both lower bounds.

First, an informal reminder of Yao's method for proving impossibility results for general randomized algorithms: Suppose that there is a fixed distribution over ``positive'' inputs (inputs that should be accepted) and a distribution over ``negative'' inputs, so that no deterministic algorithm of the prescribed type can distinguish between the two distributions. That is, suppose that for every such algorithm, the difference in the acceptance probability over both input distributions is $o(1)$. This will mean that no randomized algorithm can distinguish between these distributions as well, and hence for every possible randomized algorithm there is a positive instance and a negative instance so that it cannot be correct for both of them.

In our case an ``input'' is a distribution $\mu$ over $\{1,\ldots,n\}$, and so a ``distribution over inputs'' is in fact a distribution over distributions. To see why a distribution over distributions cannot be replaced with just a single ``averaged distribution'', consider the following example. Assume that an algorithm takes two independent samples from a distribution $\mu$ over $\{1,2\}$. If $\mu$ is with probability $\frac12$ the distribution always giving $1$, and with probability $\frac12$ the distribution always giving $2$, then the two samples will be either $(1,1)$ or $(2,2)$, each with probability $\frac12$. This can never be the case if we had used a fixed distribution for $\mu$, rather than a distribution over distributions.

What it means to be a deterministic version of our testers will be defined below; as with other settings, these result from fixing in advance the results of the coin tosses of the randomized testers.
The following are the two distributions over distributions that we will use to prove lower bounds (a third one will simply be ``pick the uniform distribution over $\{1,\ldots,n\}$ with probability $1$'').

\begin{defin}
Given a set $U\subseteq\{1,\ldots,n\}$, we define the {\em $U$-distribution} to be the uniform distribution over $U$, that is we set $p_i=1/|U|$ if $i\in U$ and $p_i=0$ otherwise.

The {\em even uniblock distribution} over distributions is defined by the following:
\begin{enumerate}
\item Uniformly choose an integer $k$ such that $\frac{1}{8}\log n\leq k\leq \frac{3}{8}\log n$.
\item Uniformly (from all possible such sets) pick a set $U\subseteq\{1,\ldots,n\}$ of size exactly $2^{2k}$.
\item The output distribution $\mu$ over $\{1,\ldots,n\}$ is the $U$-distribution (as defined above).
\end{enumerate}

The {\em odd uniblock distribution} over distributions is defined by the following:
\begin{enumerate}
\item Uniformly choose an integer $k$ such that $\frac{1}{8}\log n\leq k\leq \frac{3}{8}\log n$.
\item Uniformly (from all possible such sets) pick a set $U'\subseteq\{1,\ldots,n\}$ of size exactly $2^{2k+1}$.
\item The output distribution $\mu$ over $\{1,\ldots,n\}$ is the $U'$-distribution.
\end{enumerate}

Finally, we also identify the {\em uniform distribution} as a distribution over distributions that picks with probability $1$ the uniform distribution over $\{1,\ldots,n\}$.
\end{defin}

For these to be useful for Yao arguments, we first note their farness properties.

\begin{obs}\label{obs:far}
Any distribution over $\{1,\ldots,n\}$ that may result from the even uniblock distribution over distributions is $\frac12$-far from the uniform distribution over $\{1,\ldots,n\}$, as well as $\frac12$-far from any distribution that may result from the odd uniblock distribution over distributions.
\end{obs}

\begin{proof}
This follows directly from a variation distance calculation. Specifically, the variation distance between a uniform distribution over $U$ and (a permutation of) a uniform distribution over $V$ with $|V|\geq |U|$ (which is minimized when we make the permutation such that $U\subseteq V$) is at least $(|V|-|U|)/|V|$. In our case we always have $|V|\geq 2|U|$, and hence the lower bound.
\end{proof}

All throughout this section we consider properties that are label-invariant (such as the properties of being in the support of the distributions defined above). This allows us to simplify the analysis of our algorithms.

First, some technical definitions.

\begin{defin}
Given $A_1,\ldots,A_r\subseteq\{1,\ldots,n\}$, the {\em atoms generated by $A_1,\ldots,A_r$} are all non-empty sets of the type $\bigcap_{j=1}^rC_j$ where every $C_j$ is one of $A_j$ or $\{1,\ldots,n\}\setminus A_j$. In other words, these are the minimal (by containment) non-empty sets that can be created by boolean operations over $A_1,\ldots,A_r$. The family of all such atoms is called the {\em partition} generated by $A_1,\ldots,A_r$; when $r=0$ that partition includes the one set $\{1,\ldots,n\}$.

Given $A_1,\ldots,A_r$ and $j_1,\ldots,j_r$ where $j_i\in A_i$ for all $i$, the {\em $r$-configuration} of $j_1,\ldots,j_r$ is the information for any $1\leq l,k\leq r$ of whether $j_k\in A_l$ (or equivalently, which is the atom that contains $j_k$) and whether $j_k=j_l$.
\end{defin}




The label-invariance of all properties discussed in this section will allow us to ``simplify'' our algorithms prior to proving lower bounds. We next define a simplified version of a non-adaptive algorithm.

\begin{defin}
A {\em core non-adaptive distribution tester} is a non-adaptive distribution tester, that in its last phase bases its decision to accept or reject only on the $t(\epsilon)$-configuration of its received samples and on its internal coin tosses.
\end{defin}

For a core non-adaptive tester, fixing the values of the internal ``coins'' in advance gives a very simple deterministic counterpart (for use in Yao arguments): The algorithm now consists of a sequence of fixed sets $A_1,\ldots,A_{t(\epsilon)}$, followed by a function assigning to every possible $t(\epsilon)$-configuration a decision to accept or reject.

We note that indeed in the non-adaptive setting we only need to analyze core algorithms:

\begin{obs}\label{obs:ucore}
A non-adaptive testing algorithm for a label-invariant property can be converted to a corresponding core algorithm with the same sample complexity.
\end{obs}

\begin{proof}
We start with the original algorithm, but choose a uniformly random permutation $\sigma$ of $\{1,\ldots,n\}$ and have the algorithm act on the correspondingly permuted input distribution, rather than the original one. That is, every set $A_i$ that the algorithm conditions on is converted to $\{\sigma(k):k\in A_i\}$, while instead of $j_i$ the algorithm receives $\sigma^{-1}(j_i)$. This clearly preserves the guaranteed bounds on the error probability if the property is label-invariant.

To conclude, note that due to the random permutation, all outcomes for $j_1,\ldots,j_t$ that satisfy a given configuration are equally likely, and hence can be simulated using internal coin tosses once the configuration itself is made known to the algorithm.
\end{proof}

For an adaptive algorithm, the definition will be more complex. In fact we will need to set aside some ``external'' coin tosses, so that also the ``deterministic'' counterpart will have a probabilistic element, but it will be a manageable one.

\begin{defin}
A {\em core adaptive distribution tester} is an adaptive distribution tester, that acts as follows.
\begin{itemize}
\item In the $i$'th phase, based only on the internal coin tosses and the configuration of the sets $A_1,\ldots,A_{i-1}$ and $j_1,\ldots,j_{i-1}$, the algorithm assigns a number $k_A$ for every atom $A$ that is generated by $A_1,\ldots,A_{i-1}$, between $0$ and $|A\setminus\{j_1,\ldots,j_{i-1}\}|$, where not all such numbers are $0$. Additionally the algorithm provides $K_i\subseteq\{1,\ldots,i-1\}$.
\item A set $B_i\subseteq\{1,\ldots,n\}\setminus\{j_1,\ldots,j_{i-1}\}$ is drawn uniformly among all such sets whose intersection with every atom $A$ as above is of size $k_A$, and $A_i$ is set to $B_i\cup\{j_k:k\in K_i\}$. The random draw is done independently of prior draws and the algorithm's own internal coins, and $A_i$ is not revealed to the algorithm (however, the algorithm will be able to calculate the sizes of the atoms in the partition generated by $A_1\ldots,A_i$ using the $i-1$-configuration, and the numbers provided based on it and the internal coin tosses).
\item A sample $j_i$ is drawn according to $\mu$ conditioned over $A_i$, independently of all other draws. $j_i$ is not revealed to the algorithm, but the new $i$-configuration is revealed (in other words, the new information that the algorithm receives is whether $j_i\in A_k$ and whether $j_i=j_k$ for each $k<i$).
\item After $t(\epsilon)$ such phases, the algorithm bases its decision to accept or reject only on the $t$-configuration of its received samples and on its internal coin tosses.
\end{itemize}
\end{defin}

Note that also a ``deterministic'' version of the above algorithm acts randomly, but only in a somewhat ``oblivious'' manner. The sets $A_i$ will still be drawn at random, but the decisions that the algorithm is allowed to make about them (through the $k_A$ numbers and the $K_i$ sets) as well as the final decision whether to accept or reject will all be deterministic. This is since a deterministic version fixes the algorithm's internal coins and only them.

Also for adaptive algorithms we need to analyze only the respective core algorithms.

\begin{obs}\label{obs:acore}
An adaptive testing algorithm for a label-invariant property can be converted to a corresponding core algorithm with the same sample complexity.
\end{obs}

\begin{proof}
Again we use a uniformly random permutation $\sigma$ of $\{1,\ldots,n\}$. Regardless of how the original set $A_i$ was chosen, now it will be chosen uniformly at random among all sets satisfying the same intersection sizes with the atoms of the partition generated by $A_1,\ldots,A_{i-1}$ and the same membership relations with $j_1,\ldots,j_{i-1}$. Hence the use of a uniformly drawn set based on the $k_A$ numbers and $K_i$ is justified, and since $\sigma$ is not revealed to the algorithm, the particular resulting set $A_i$ is not revealed.

Also, the probability for a particular value of $j_i$ now can depend only on the resulting $i$-configuration, and hence it is sufficient to reveal only the configuration to the algorithm -- the algorithm can then use internal coin tosses to simulate the actual value of $j_i$ (uniformly drawing it from all values satisfying the same configuration). The same goes for the decision whether to accept or reject in the end.

To further illustrate the last point, note that the analysis does not change even if we assume that at every phase, after choosing $A_i$ we also draw a new random permutation, chosen uniformly at random among all those that preserve $j_1,\ldots,j_{i-1}$ and the atoms of $A_1,\ldots,A_i$ (but can ``reshuffle'' each atom internally). Then the ``position inside its atom'' of $j_i$ will be completely uniform among those of the same configuration (if the configuration makes it equal to a previous $j_k$ then there is only one choice for $j_i$ anyway).
\end{proof}

\subsection{Uniformity has no constant sample non-adaptive test}\label{sec:unilb}
\newcommand{\denom}[1]{2^{\sqrt{\log n}/4}}

\begin{theorem}\label{thm:lb_uni}
 Testing uniformity requires at least $\Omega(\log\log n)$ non-adaptive
  conditional samples (for some fixed $\epsilon$).
\end{theorem}

\begin{proof}
This follows from Lemma \ref{lem:lb_uni} below.
\end{proof}

To prove this lower bound, we show that for any fixed $t$ and large enough $n$, no deterministic non-adaptive algorithm can distinguish with probability $\frac13$ between the case where the input distribution is the uniform one (with probability $1$), and the case where the input distribution is drawn according to the even uniblock distribution over distributions. Recall that such a deterministic algorithm is in fact given by fixed sets $A_1,\ldots,A_t \subseteq \{1,\ldots,n\}$ and a fixed acceptance criteria based on the $t$-configuration of the obtained samples (to see this, take a core non-adaptive testing algorithm and arbitrarily fix its internal coins).

We now analyze the performance of a deterministic non-adaptive tester against the even uniblock distribution. Asymptotic expressions are for an increasing $n$.

\begin{defin}
We call a set $A\subseteq\{1,\ldots,n\}$ {\em large} if $|A|>n2^{\sqrt{\log n}}/|U|$, where $U$ is the set chosen in the construction of the even uniblock distribution. We call $A$ {\em small} if $|A|<n2^{-\sqrt{\log n}}/|U|$.
\end{defin}

\begin{lemma}\label{lem:smlgn}
With probability at least $1-\frac{2^{t+2}}{\sqrt{\log n}}$ over the choice of $U$, all atoms in the partition generated by $A_1,\ldots,A_t$ are either large or small.
\end{lemma}

\begin{proof}
There are at most $2^t$ atoms. An atom $A$ is neither large nor small if $n2^{-\sqrt{\log n}}\leq |A||U|\leq n2^{\sqrt{\log n}}$. $|U|=2^{2k}$ where $\frac{1}{8}\log n\leq k\leq \frac{3}{8}\log n$ is chosen uniformly. Therefore, for a fixed $A$, there are at most $\sqrt{\log n}$ values of $k$ which will make it neither large nor small. Since the range of $k$ is of size $\frac14\log n$, we get that with probability at most $\frac{4}{\sqrt{\log n}}$ $A$ is neither large nor small. Taking the union bound over all atoms gives the statement of the lemma.
\end{proof}

\begin{lemma}\label{lem:disjn}
With probability at least $1-2^{t-\sqrt{\log n}}$, no small atom intersects $U$.
\end{lemma}

\begin{proof}
Given a fixed $k$, for any small set $A$ the probability of it intersecting $U$ is clearly bounded by $2^{-\sqrt{\log n}}$. We can now conclude the proof by union-bounding over all small atoms, whose number is bounded by $2^t$.
\end{proof}

\begin{lemma}\label{lem:repn}
With probability $1-\exp\left(t-t^2\cdot 2^{\sqrt{\log n}/2-1}\right)$, for every large atom $A$, we have $|A\cap U|=\left(1\pm \frac{t}{\denom{n}}\right)|A|\cdot|U|/n$.
\end{lemma}

\begin{proof}
This is by a large deviation inequality followed by a union bound over all atoms. Note first that if instead of $U$ we had a uniformly random sequence $u_1,\ldots,u_{2^{2k}}$ (chosen with possible repetitions), then this would have been covered by Lemma \ref{lem:atomcon}. However, $U$ is a random set of fixed size instead. For this we appeal to Section 6 of \cite{Hypergeom}, where it is proved that moving from a Binomial to a Hypergeometric distribution (which corresponds to choosing the set $U$ with the fixed size) only makes the distribution more concentrated. The rest follows by the fact that $A$ is large.
\end{proof}

Now we can take $t\leq\frac14\log\log n$ and put forth the following lemma, which implies that the even uniblock distribution over distributions is indeed indistinguishable by a deterministic non-adaptive core algorithm from the uniform distribution using only $t$ samples.

\begin{lemma}\label{lem:ueindis}
For $t\leq\frac14\log\log n$, with probability $1-o(1)$, the distribution over $\{1,\ldots,n\}$ obtained from the even uniblock distribution over distributions, is such that the resulting distribution over the configurations of $j_1,\ldots,j_t$ is $o(1)$-close in the variation distance to the distribution over configurations resulting from the uniform distribution over $\{1,\ldots,n\}$.
\end{lemma}

\begin{proof}
With probability $1-o(1)$ all of the events in Lemmas \ref{lem:smlgn}, \ref{lem:disjn} and \ref{lem:repn} occur. We prove that in this case the two distributions over configurations are $o(1)$-close. Recall that the uniform distribution over the set $U$ (resulting from the uniblock distribution) is called the {\em $U$-distribution}. The lemma follows from the following:
\begin{itemize}
\item A sample taken from a set $A_i$ that contains only small atoms will be uniform from this set (and independent of all others), both for the uniform distribution and the $U$-distribution. For the $U$-distribution it follows from $U$ not intersecting $A_i$ at all (recall that in our model, a conditional sample with a set of empty weight returns a uniformly random element from that set).
\item A sample taken from a set $A_i$ that contains some large atom will not be identical to any other sample with probability $1-o(1)$ for both distributions. This follows from the birthday paradox: Setting $A$ to be the large atom contained in $A_i$, recall that $|A\cap U|=\left(1\pm \frac{\log\log n/4}{\denom{n}}\right)|A|\cdot|U|/n$. This quantity is $\omega(\log^2\log n)$. Thus for a fixed $i$ the probability for a collision with any other $j$ is $o(1/\log\log n)$ (regardless of whether $A_j$ contains a large atom), and hence with probability $1-o(1)$ there will be no collision for any $i$ for which $A_i$ contains a large atom.
\item For a set $A_i$ containing a large atom, the distribution over the algebra of the events $j_i\in A_k$ (which corresponds to the distribution over the atom in the partition generated by $A_1,\ldots,A_t$ containing $j_i$) are $o(1)$ close for both distributions. To show this we analyze every atom $A$ generated by $A_1,\ldots,A_t$ that is contained in $A_i$ separately. If $A$ is small, then for the uniform distribution, $j_i$ will not be in it with probability $1-o(1)$ (a small atom is in particular of size $o(|A_i|)$ since $A_i$ contains a large atom as well), while for the $U$-distribution this is with probability $1$ (recall that we conditioned on the event of $U$ not intersecting any small atom). If $A$ is large, then we have $|A\cap U|=\left(1\pm \frac{\log\log (n)/4}{\denom{n}}\right)|A|\cdot|U|/n$, implying that the probabilities for $j_i\in A$ for the $U$-distribution and the uniform one are only $o(1)$ apart.
\end{itemize}
The items above allow us to conclude the proof. They mean that for both the $|U|$-distribution (conditioned on the events in Lemmas \ref{lem:smlgn}, \ref{lem:disjn} and \ref{lem:repn}) and the uniform distribution, the resulting distributions over configurations are $o(1)$-close to the one resulting by setting the following:
\begin{enumerate}
\item For every $i$ for which $A_i$ contains only small atoms, uniformly pick $j_i\in A_i$ independently of all other random choices; write down the equalities between these samples and the atoms to which these samples belong.
\item For every $i$ for which $A_i$ contains a large atom, write $j_i$ as having no collisions with any other sample; then pick the atom containing $j_i$ from all atoms contained in $A_i$ according to their relative sizes, in a manner independent of all other random choices.
\end{enumerate}
\end{proof}

Lemma \ref{lem:ueindis} allows us to conclude the argument by Yao's method.

\begin{lemma}\label{lem:lb_uni}
All non-adaptive algorithms taking $t\leq\frac14\log\log n$ conditional samples will fail to distinguish the uniform distribution from the even uniblock distribution over distributions (which are all $\frac12$-far from uniform) with any probability more than $o(1)$.
\end{lemma}

\begin{proof}
By Observation \ref{obs:ucore} it is enough to consider core non-adaptive algorithms, and by Yao's argument it is enough to consider deterministic ones.

For any deterministic non-adaptive core algorithm (characterized by $A_1,\ldots,A_t$ and a function assigning a decision to every possible configuration), the even uniblock distribution with probability $1-o(1)$ will choose a $U$-distribution, which in turn will induce a distribution over configurations that is $o(1)$-close to that induced by the uniform distribution over $\{1,\ldots,n\}$. This means that if we look at the distribution over configurations caused by the even uniblock distribution over distributions itself, it will also be $o(1)$-close to the one induced by the uniform distribution.
Therefore the acceptance probabilities of the algorithm for both distributions over distributions are $o(1)$-close.
\end{proof}

It would be interesting to see if the bound on the number of samples can be made into a power of $\log n$, possibly by analyzing the sets $A_i$ by themselves rather than through their generated partition.

\subsection{A label-invariant property with no constant sample adaptive test}

\begin{theorem} \label{thm:lb_label_invar}
There exists a label invariant property such that any adaptive testing algorithm for it must use at least $\Omega(\sqrt{\log\log n})$ conditional samples (for some $\epsilon$).
\end{theorem}

\begin{proof}
This follows from Lemma \ref{lem:lb_label_invar} below.
\end{proof}

The property will be that of the distribution being the possible result of the even uniblock distribution over distributions. In other words, it is the property of being equal to the $U$-distribution over some set $U$ of size $2^{2k}$ for some $\frac18\log n\leq k\leq\frac38\log n$.

We show that no ``deterministic'' adaptive core algorithm can distinguish between the even and odd uniblock distributions using $o(\sqrt{\log\log n})$ samples, while by Observation \ref{obs:far} a proper $\frac12$-test must distinguish between these. Considering such algorithms, we first note that they can be represented by decision trees, where each node of height $i$ corresponds to an $i-1$-configuration of the samples made so far. An internal node describes a new sample, through the numbers $k_A$ provided for every atom $A$ of $A_1,\ldots,A_i$ (where the atoms are labeled by their operations, as the $A_i$ themselves are not revealed to the algorithm), and the set $K_i$. All these parameters can be different for different nodes of height $i$. A leaf is labeled with an accept or reject decision.

The basic ideas of the analysis are similar to those of the previous subsection, but the analysis itself is more complex because we have to consider the ``partition generated by the samples so far'' in every step of the algorithm. The first thing to note is that there are not too many nodes in the decision tree.

\begin{obs}\label{obs:notmuch}
The number of nodes in a decision tree corresponding to a $t$-sample algorithm is less than $t2^{2t^2}$.
\end{obs}

\begin{proof}
A configuration can be described by assigning each of the $i$ samples with a vector of length $2i$, indicating which sets do they belong to and which of the other samples are they equal to. This gives an $i\times 2i$ binary matrix, where every possible $i$-configuration for $i$ samples corresponds to some such matrix. That gives us at most $2^{2i^2}$ possible $i$-configurations. Summing for all $i\leq t$ gives the bound in the statement.
\end{proof}

From now on we will always assume that $n$ is larger than an appropriate fixed constant.
For the analysis, we consider two input distributions as being drawn at once, one according to the even uniblock distribution and the other according to the odd uniblock distribution. We first choose $\frac{1}{8}\log n\leq k\leq \frac{3}{8}\log n$ uniformly at random, and then uniformly choose a set $U$ of size $2^{2k}$ and a set $U'$ of size $2^{2k+1}$. We then set $\mu$ to be the $U$-distribution and $\mu'$ to be the $U'$-distribution.

We will now show that the fixed decision tree accepts with almost the same probability when given either $\mu$ or $\mu'$, which will allow us to conclude the proof using Yao's argument. We start with a notion of ``large'' and ``small'' similar to the one used for non-adaptive algorithms, only here we need it for the numbers themselves.

\begin{defin}
We call a number $b$ {\em large} with respect to $U$ if $b>n2^{\sqrt{\log n}}/|U|$. We call $b$ {\em small} with respect to $U$ if $b<n2^{-\sqrt{\log n}}/|U|$. We make the analogous definitions with respect to $U'$.
\end{defin}

\begin{lemma}\label{lem:smlga}
With probability at least $1-\frac{t2^{3t^2+2}}{\sqrt{\log n}}$, all ``$k_A$'' numbers appearing in the decision tree are either small with respect to both $U$ and $U'$, or large with respect to both $U$ and $U'$.
\end{lemma}

\begin{proof}
By Observation \ref{obs:notmuch} the total of different ``$k_A$'' numbers is no more than $t2^{3t^2}$ (the number of nodes times $2^t$ -- the bound on the size of the partition generated by $A_1,\ldots,A_i$ in every node).
We can conclude similarly to the proof of Lemma \ref{lem:smlgn} that since $|U|$ and $|U'|$ differ by a factor of $2$, there are at most $\sqrt{\log n}$ values of $k$ for which some fixed number $k_A$ will not be either large with respect to both or small with respect to both. The bound in the statement then follows by union bound.
\end{proof}

From now on we assume that the event of Lemma \ref{lem:smlga} has occurred, and fix $k$ (that is, the following will hold not only for the entire distributions, but also for the conditioning on every specific $k$ for which the event of Lemma \ref{lem:smlga} is satisfied). The following lemma is analogous to the non-adaptive counterparts Lemma \ref{lem:disjn} and Lemma \ref{lem:repn}, but here it is proved by induction for every node that is reached while running the decision tree over the distribution drawn according to either $\mu$ or $\mu'$, where the inductive argument requires both statements to hold. This lemma will essentially be used as a warm-up, since the final proof will refer to the proof and not just the statement of the lemma.

\begin{lemma}\label{lem:disjrepa}
Assuming that $t\leq \sqrt{\frac{1}{32}\log\log n}$, and conditioned on that the events of Lemma \ref{lem:smlga} have occurred, for every $1\leq i\leq t$, with probability at least $1-\frac{2^{t+1}}{\sqrt{\log n}}$, the following occur.
\begin{itemize}
\item All small atoms in the partition generated by $A_1,\ldots,A_i$ contain no members of either $U$ or $U'$ outside (possibly) $\{j_1,\ldots,j_{i-1}\}$.
\item For every large atom $B$ in the partition generated by $A_1,\ldots,A_i$, we have both $|B\cap U|=\left(1\pm \frac{i}{\denom{n}}\right)|B|\cdot|U|/n$ and $|B\cap U'|=\left(1\pm \frac{i}{\denom{n}}\right)|B|\cdot|U'|/n$.
\end{itemize}
\end{lemma}

\begin{proof}
We shall prove the lemma not only conditioned on the event of Lemma \ref{lem:smlga}, but also conditioned on any fixed $|U|$ (and $|U'|=2|U|$) for which Lemma \ref{lem:smlga} is satisfied. We assume by induction that this occurs for the atoms in the partition generated by $A_1,\ldots,A_{i-1}$ with probability at least $1-\frac{2^i}{\sqrt{\log n}}$, and prove it for $A_1,\ldots,A_i$ with probability at least $1-\frac{2^{i+1}}{\sqrt{\log n}}$. Recall that the way $A_i$ is generated, the algorithm in fact specifies how many members of it will appear in $A\setminus\{j_1,\ldots,j_{i-1}\}$ for every atom $A$ of the partition generated by $A_1,\ldots,A_{i-1}$ (while specifying exactly which of $j_1,\ldots,j_{i-1}$ will appear in it), and then the actual set is drawn uniformly at random from those that satisfy the specifications.

We show the conclusion of the lemma to hold even if $U$ and $U'$ are held fixed (as long as they satisfy the induction hypothesis and their sizes satisfy the assertion of Lemma \ref{lem:smlga}). Let $B$ be an atom of $A_1,\ldots,A_i$ and let $A$ be the atom of $A_1,\ldots,A_{i-1}$ so that $B\subseteq A$. We have several cases to consider, conditioned on the fact that the event in the statement does occur for $i-1$.
\begin{itemize}
\item If $A$ is small, then so is $B$. By the induction hypothesis $A\setminus\{j_1,\ldots,j_{i-1}\}$ has no members of $U$ or $U'$, and hence so does $B$. This happens with (conditional) probability $1$.
\item If $A$ is large but $B$ is small, by the induction hypothesis both $|A\cap U|=\left(1\pm \frac{(i-1)}{\denom{n}}\right)|A|\cdot|U|/n$ and $|A\cap U'|=\left(1\pm \frac{(i-1)}{\denom{n}}\right)|A|\cdot|U'|/n$. When this happens, as $B\setminus\{j_1,\ldots,j_{i-1}\}$ is in fact chosen uniformly from all subsets of $A\setminus\{j_1,\ldots,j_{i-1}\}$ of the same size (either $k_A$ or $|A\setminus\{j_1,\ldots,j_{i-1}\}|-k_A$), and since $B$ is small, we can use a union bound to see that no member of either $U$ or $U'$ is taken into $B$, with probability at least $1-2^{1-\sqrt{\log n}}$.
\item If $B$ is large (and hence so is $A$), then again by the induction hypothesis both $|A\cap U|=\left(1\pm \frac{(i-1)}{\denom{n}}\right)|A|\cdot|U|/n$ and $|A\cap U'|=\left(1\pm \frac{(i-1)}{\denom{n}}\right)|A|\cdot|U'|/n$. We also note that since $B$ is large we have in particular $t\leq \frac{1/2}{\denom{n}}|B|$. We can now use a large deviation inequality (as in Lemma \ref{lem:repn}) to conclude the bounds for $|B\cap U|$ and $|B\cap U'|$ with probability $1-2\exp(-2^{\sqrt{\log n}/2-2})$.
\end{itemize}
Thus in all cases the statement will not hold with probability at most $\frac{1}{\sqrt{\log n}}$ for $n$ large enough. By taking the union bound over all possibilities for $B$ (up to $2^i$ events in total) we get that with probability $1-\frac{2^{i}}{\sqrt{\log n}}$ the statement of the lemma holds for $A_1,\ldots,A_i$, conditioned on the event occurring for $A_1,\ldots,A_{i-1}$.  A union bound with the event of the induction hypothesis happening for $A_1,\ldots,A_{i-1}$ gives the required probability bound.
\end{proof}

We now prove the lemma showing the indistinguishability of $\mu$ from $\mu'$ whenever $t\leq \sqrt{\frac{1}{32}\log\log n}$, conditioned on the event of Lemma \ref{lem:smlga}. We assume without loss of generality that the decision tree of the algorithm is full and balanced, which means that the algorithm will always take $t$ samples even if its output was already determined before they were taken.

\begin{lemma}\label{lem:nodereach}
Assuming that $t\leq \sqrt{\frac{1}{32}\log\log n}$ and that the event of Lemma \ref{lem:smlga} has occurred, consider the resulting distributions of which of the leaves of the algorithm was reached. These two distributions, under $\mu$ compared to under $\mu'$, are at most $\frac{2^{3t+1}}{\sqrt{\log n}}$ apart from each other.
\end{lemma}

\begin{proof}
The proof is reminiscent of the proof of Lemma \ref{lem:ueindis} above, but requires more cases to be considered, as well as induction over the height of the nodes under consideration. Denoting this height by $i$, we shall prove by induction that the distributions over which of the height $i$ nodes was reached, under $\mu$ compared to $\mu'$, are only at most $1-\frac{2^{3i+1}}{\sqrt{\log n}}$ apart from each other.

We shall use the induction hypothesis that the corresponding distributions of the node of height $i-1$ (the parent of the node that we consider now) are at most $1-\frac{2^{3i-2}}{\sqrt{\log n}}$ apart, and then show that the variation distance between the distributions determining the transition from a particular parent to a child node is no more than $\frac{2^{3i}}{\sqrt{\log n}}$, which when added to the difference in the distributions over the parent nodes gives required bound.


The full induction hypothesis will include not only the bound on the distributions of the parent nodes, but also a host of other assumptions, that we prove along to occur with probability at least $1-\frac{2^{3i+1}}{\sqrt{\log n}}$. In particular, instead of using the statement of Lemma \ref{lem:disjrepa}, we essentially re-prove it here. So the induction hypothesis also includes that all of the events proved during the inductive proof of Lemma \ref{lem:disjrepa} hold here with respect to $A_1,\ldots,A_{i-1}$. Also, as in the proof of Lemma \ref{lem:disjrepa}, the conditional probability of them not holding for $A_1,\ldots,A_i$ is at most $\frac{2^i}{\sqrt{\log n}}$ (by the union bound done there for every atom generated by $A_1,\ldots,A_i$ of the event of the hypothesis failing for any single atom $A$). Therefore, we assume that additionally the inductive hypothesis used in the proof of Lemma \ref{lem:disjrepa} has occurred for $A_1,\ldots,A_i$, and prove that with probability at least $1-\frac{2^{2i}}{\sqrt{\log n}}$ all other assertions of the inductive hypothesis occur as well as that the variation distance between the distributions over the choice of the child node is at most $\frac{2^{2i}}{\sqrt{\log n}}$. By a union bound argument (and for the variation distance, a ``common large probability event'' argument), this will give us the $1-\frac{2^{3i}}{\sqrt{\log n}}$ bound that we need for the induction. Recall that the choice of child node depends deterministically on the question of which atom of $A_1,\ldots,A_i$ contains the obtained sample $j_i$, so in fact we will bound the distance between the distributions of the atom in which $j_i$ has landed.


Additionally, we define by induction over $i$ the following notion: An index $i$ is called {\em smallish} if all the ``$k_A$'' numbers relating to it are small, and additionally $K_i$ contains only smallish indexes (recall that $K_i\subseteq\{1,\ldots,i-1\}$). A final addition to our induction hypothesis is that with probability at least $1-\frac{2^{3i-2}}{\sqrt{\log n}}$, in addition to all our other assertions, the following occur for every $i'<i$.
\begin{itemize}
\item The sample $j_{i'}$ is in $U$ or respectively $U'$ if and only if $i'$ is not smallish (note that the assignment of smallish indexes depends on the parent node).
\item If $i'$ is not smallish but all its corresponding ``$k_A$'' numbers are small, then $j_{i'}$ is equal to some $j_l$ where $l$ is a non-smallish index smaller than $i'$.
\item If there exists a large ``$k_A$'' number for $i'$, then $j_{i'}$ is not equal to $j_l$ for any $l<i'$, and additionally $j_{i'}$ lies in some atom $A'$ for which the corresponding $k_{A'}$ is not small (it is allowed that $A'=A$).
\end{itemize}

We now work for every possible parent node of height $i-1$ separately. Note that we restrict our attention to nodes whose corresponding $(i-1)$-configurations satisfy the induction hypothesis. Recall that we assume that the induction hypothesis in the proof of Lemma \ref{lem:disjrepa} has occurred for $A_1,\ldots,A_i$, and aim for a $\frac{2^{2i}}{\sqrt{\log n}}$ ``failure probability'' bound. We separate to cases according to the nature of $A_1,\ldots,A_i$.
\begin{itemize}
\item A sample taken from a set $A_i$, where $i$ is smallish, will be uniform and independent of other samples, for both the $U$-distribution and the $U'$-distribution. Moreover, this $j_i$ in itself will not be a member of $U$ or respectively $U'$. This is since $A_i\setminus\{j_k:k\in K_i\}$ does not intersect $U$ or $U'$, together with the induction hypothesis for $\{j_k:k\in K_i\}$ (so also $A_i$ does not intersect $U$ or $U'$). So conditioned on the entire induction hypothesis for $i-1$ and the hypothesis in the proof of Lemma \ref{lem:disjrepa} for $A_1,\ldots,A_i$, all assertions for $i$ will occur with probability $1$, and the distributions for selecting the height $i$ node given this particular parent node are identical under either $\mu$ or $\mu'$.
\item A sample taken from a set $A_i$, where the $k_A$ numbers are all small but $i$ is not smallish, will be a member of $U$ or respectively $U'$, chosen uniformly (and independently) from $\{j_k:k\in K'_i\}$, where $K'_i$ denotes the (non-empty) set of all non-smallish indexes in $K_i$. This is since $\{j_k:k\in K'_i\}$ is exactly the set of members of $U$ or respectively of $U'$ in $A_i$ (by the hypothesis for $A_1,\ldots,A_i$ there will be no member of $U$ or $U'$ in $A_i\setminus\{j_k:k\in K_i\}$, and the rest follows from the induction hypothesis concerning smallish indexes). Again the assertions for $i$ follow with probability $1$ (conditioned on the above hypotheses), and the distributions for selecting the height $i$ node are identical.
\item If a sample is taken from $A_i$ where at least one of the $k_A$ numbers is not small, then the following occur.
\begin{itemize}
\item Since $A_i$ in particular contains the atom $A$, and both $|A\cap U|=\left(1\pm \frac{i}{\denom{n}}\right)|A|\cdot|U|/n$ and $|A\cap U'|=\left(1\pm \frac{i}{\denom{n}}\right)|A|\cdot|U'|/n$ by the assertion over $A_1,\ldots,A_i$ relating to Lemma \ref{lem:disjrepa}, we note that in particular $i=o(\frac{1}{\sqrt{\log n}}|A_i\cap U|)$ and $i=o(\frac{1}{\sqrt{\log n}}|A_i\cap U'|)$, so with probability less than $\frac{1}{\sqrt{\log n}}$ (for $n$ larger than some constant) we will get under either $\mu$ or $\mu'$ a sample that is identical to a prior one.
\item By the assertion over $A_1,\ldots,A_i$, an atom $B$ inside $A_i$ for which the corresponding $k_B$ is small will not contain a member of $U$ or $U'$, and so $j_i$ will not be in such an atom (in the preceding item we have already established that there are members of $U$ and respectively $U'$ in $A_i$).
\item By the assertion over $A_1,\ldots,A_i$, for every large atom $B$ inside $A_i$ we have both $|B\cap U|=\left(1\pm \frac{i}{\denom{n}}\right)|B|\cdot|U|/n$ and $|B\cap U'|=\left(1\pm \frac{i}{\denom{n}}\right)|B|\cdot|U'|/n$, implying that $\frac{|B\cap U|}{|U|}=\left(1\pm \frac{i}{2^{\sqrt{\log n}/5}}\right)\frac{|B\cap U'|}{|U'|}$ (for large enough $n$). Also, every small atom $C$ inside $A_i$ contains no members of $U$ or $U'$, so summing over all atoms of $A_i$ we obtain $\frac{|A_i\cap U|}{|U|}=\left(1\pm \frac{i}{2^{\sqrt{\log n}/5}}\right)\frac{|A_i\cap U'|}{|U'|}$, and thus for every atom $B$ of $A_i$ (large or small) we finally have $\frac{|B\cap U|}{|A_i\cap U|}=\left(1\pm \frac{i}{2^{\sqrt{\log n}/6}}\right)\frac{|B\cap U'|}{|A_i\cap U'|}$ (for small atoms both sides are zero).

Te final thing to note is that $\frac{|B\cap U|}{|A_i\cap U|}$ and respectively $\frac{|B\cap U'|}{|A_i\cap U'|}$ equal the probabilities of obtaining a sample from $B$ under $\mu$ and respectively $\mu'$. Summing over all atoms contained in $A_i$ (of which there are $2^{i-1}$) we obtain a difference over these distributions that is bounded by $\frac{2^i}{\sqrt{\log n}}$, which satisfies the requirements (also after conditioning on that the events related to the rest of the induction hypothesis have occurred).
\end{itemize}
\end{itemize}
Having covered all cases, this completes the proof that the inductive hypothesis follows to $i$, and thus the proof of the lemma is complete.
\end{proof}

Now we can conclude the argument by Yao's method to prove the following lemma that implies the theorem.

\begin{lemma}\label{lem:lb_label_invar}
All adaptive algorithms taking $t\leq \sqrt{\frac{1}{32}\log\log n}$ conditional samples will fail to distinguish the even uniblock distribution over distributions from the odd one (whose outcomes are always $\frac12$-far from those of the even distribution) with any probability more than $o(1)$.
\end{lemma}

\begin{proof}
By Observation \ref{obs:acore} it is enough to consider only core adaptive algorithms, and then by Yao's argument it is enough to consider ``deterministic'' ones (the quote marks are because the external coin tosses are retained as per the definitions above). We now consider the decision tree of such an algorithm, and feed to it either $\mu$ or $\mu'$ that are drawn as per the definition above. With probability at least $1-\frac{t2^{3t^2+2}}{\sqrt{\log n}}=1-o(1)$ the event of Lemma \ref{lem:smlga} has occurred, and conditioned on this event (or even if we condition on particular $U$ and $U'$), Lemma \ref{lem:nodereach} provides that the variation distance between the resulting distributions over the leafs is at most $\frac{2^{3t+1}}{\sqrt{\log n}}=o(1)$. In particular this bounds the difference between the (conditional) probabilities of the event of reaching an accepting leaf of the algorithm.

Since we have an $o(1)$ difference when conditioned on a $1-o(1)$ probability event, we also have an $o(1)$ difference on the unconditional probability of reaching an accepting leaf under $\mu$ compared to $\mu'$. This means that the algorithm cannot distinguish between the two corresponding distributions over distributions.
\end{proof}

\section{A lower bound for testing general properties of distributions} \label{sec:lb}

\setcounter{theorem}{0}

For properties that are not required to be label-invariant, near-maximal non-testability could happen also when conditional samples are allowed.

\begin{theorem}\label{thm:genlb}
Some properties of distributions on $[n]$ require $\Omega(n)$ conditional samples to test (adaptive or not).
\end{theorem}

\begin{proof}
We assume that $n$ is even. We reduce the problem of testing general $n/2$-bit binary string properties $P \subseteq \zo^{n/2}$ to the problem of testing properties of distributions over $[n]$ using conditional samples, through Lemma \ref{lem:genreduce} below. Then the lower bound of the theorem follows by the existence of hard properties $P \subseteq \zo^{n/2}$ that require $\Omega(n)$ queries to test, such as the original one of \cite{ggr} or the one of \cite{Ben-SassonHR05}.
\end{proof}


The reduction proved in Lemma \ref{lem:genreduce} is probabilistic in nature, succeeding with probability $1-o(1)$ (which is sufficient for the hardness arguments to work), and only incurs an additional $O(1)$ factor in the query complexity. This means that every conditional sample made by the distribution tester is translated into (expected) $O(1)$ queries to the input binary string $x \in \zo^{n/2}$. The rest of this section is devoted to its proof.


\subsection{The reduction lemma}
We start with a few definitions.
A string $y \in \zo^n$ is {\em balanced} if it has the same number of $0$s and $1$s (in particular we assume here that $n$ is even). For $x\in \zo^{n/2}$, let $b(x) \in \zo^n$ be the string obtained by concatenating $x$ with its bitwise complement (in which every original bit of $x$ is flipped). Clearly $b(x)$ is balanced for all $x$.

For a property $P \subseteq \zo^{n/2}$, define $b(P) \subseteq \zo^n$ as $b(P)\triangleq \{b(x): x \in P\}$.

\begin{obs}
For all $x,y \in \zo^{n/2}$, $d(x,y)=d(b(x),b(y))$.
\end{obs}
\begin{proof}
Follows from the fact that if $x$ and $y$ differ in $d(x,y)\cdot\frac{n}{2}$ entries, then $b(x)$ and $b(y)$ differ in $d(x,y)\cdot{n}$ entries.
\end{proof}

\begin{obs}\label{obs:red}
For all $P$ and $\epsilon >0$, $\epsilon$-testing $b(P)$ requires at least as many queries as $\epsilon$-testing $P$.
\end{obs}

\begin{proof}
This is since we can simulate the tester for $b(P)$ also for a non balanced string $x\in \zo^{n/2}$, where a query for an index $i\leq n/2$ would return $x_i$, and for $i>n/2$ the query would return $1-x_{i-n/2}$.
\end{proof}

Next, for every balanced string $x\in \{0,1\}^n$ we define a
distribution $\mu_x$ on $[n]$ as follows:
\begin{itemize}
\item If $x_i = 0$ then $\mu_x(i) = \frac{1}{2n}$;
\item if $x_i = 1$ then $\mu_x(i) = \frac{3}{2n}$.
\end{itemize}
Note that since $x$ is balanced $\mu_x$ is indeed a distribution as $\sum_{i=1}^n \mu_x(i) =
1$.

Extending this definition further, for every property $P \subseteq \zo^{n/2}$ we define a property $\PP_P$ of distributions over $[n]$ as follows: $\PP_P \triangleq \{\mu_{x} : x \in b(P)\}.$

\begin{obs}\label{obs:dist}
For all $x,y \in \zo^{n/2}$, $d(b(x),b(y))=2 \cdot d(\mu_{b(x)},\mu_{b(y)})$,
 where the first distance refers to the normalized Hamming distance between binary strings, and the second is the variation distance between distributions.
\end{obs}

\begin{proof}
This follows from direct calculation.
\end{proof}

\begin{lemma}\label{lem:genreduce}
For all $P$ and $\epsilon >0$, if $\epsilon$-testing $P$ with success probability $3/5$ requires at least $q$ queries, then $\epsilon/2$-testing $\PP_P$ with success probability $2/3$ requires at least $q/100$ conditional samples.
\end{lemma}

\begin{proof}
By Observation \ref{obs:dist}, for all $x \in \zo^{n/2}$, if $x \in P$ then $\mu_{b(x)} \in \PP_P$, and if $d(x,P)>\epsilon$ then $d(\mu_{b(x)}, \PP_P) > \epsilon/2$.
Now we show how to reduce the task of testing $P$ to testing $\PP_P$.
Let $T$ be a tester for $\PP_P$ making at most $q/100$ conditional samples. Given an oracle access to the input string $x \in \zo^{n/2}$, which is to be tested for membership in $P$, we simulate each conditional sample $\emptyset \ne Q \subseteq [n]$ to $\mu_{b(x)}$ made by $T$ as follows:

\paragraph{Sampler}
\begin{enumerate}
\item Pick $i \in Q$ uniformly at random. If $i \leq n/2$ query $x_i$ and set $v_i\gets x_i$. Else, query $x_{i-n/2}$ and set $v_i\gets 1-x_{i-n/2}$. \label{st:rep}
\item If $v_i=1$, output $i$.
\item Else, with probability $1/3$ output $i$, and with the remaining probability go to Step \ref{st:rep}.
\end{enumerate}
It is clear that whenever Sampler outputs $i$ with $v_i=1$, then $i$ is distributed uniformly among all indices $\{j \in Q: v_j=1\}$. The same is true for $i$ such that $v_i=0$. So, to show that Sampler simulates conditional samples correctly, it remains to prove that the ratio between the probability of outputting $i$ with $v_i=1$ and the probability of outputting $i$ with $v_i=0$ is correct.

Let $q_1\triangleq |\{i \in Q : v_i=1\}|$ and $q_0\triangleq |\{i \in Q : v_i=0\}|$.
According to our distribution $\mu_b(x)$, the distribution of indices in $Q$ corresponding to the conditional sample is as follows:
\begin{itemize}
\item $\Pr[i]=\frac{3}{3q_1 + q_0}$ if $v_i=1$.
\item $\Pr[i]=\frac{1}{3q_1 + q_0}$ if $v_i=0$.
\end{itemize}
In particular, the probability of selecting $i$ such that $v_i=1$ is $3q_1/q_0$ times the probability of selecting $i$ with $v_i=0$.

Let us now analyze what is the probability with which Sampler outputs (eventually) an index $i \in Q$ with $v_i=1$, and with $v_i=0$, respectively.
At every round, an index $i$ with $v_i=1$ is output with probability $\frac{q_1}{q_1+q_0}$, and an index $i$ with $v_i=0$ is output with probability $\frac{q_0}{3(q_1+q_0)}$. With the remaining probability (of $\frac{2q_0}{3(q_1+q_0)}$) no index is output, and the process repeats independently of all previous rounds. Hence the ratio of the probability of outputting $i$ such that $v_i=1$ to the probability of outputting $i$ with $v_i=0$ is $3q_1/q_0$, as required. Note also that the expected number of rounds (and so queries to $x$) per one execution of Sampler is $(1-\frac{2q_0}{3(q_1+q_0)})^{-1} \leq 3$.

The last ingredient in the reduction is a total-query counter, that makes sure that the number of queries to $x$ does not exceed $q$ (the lower bound). If so, the reduction fails. Since Sampler is called at most $q/100$ times (the query complexity of $T$), a $3/100<1/15$ bound on the failure probability follows by Markov's inequality, and we are done (the bound on the success probability follows even if we assume that the distribution tester ``magically'' guesses the correct answer whenever the reduction to the string property fails).
\end{proof}










\bibliographystyle{plain}
\bibliography{cond}

\end{document}